\documentclass[10pt,onecolumn,twoside]{article} 
    
\usepackage[margin = 1.0in]{geometry}
\usepackage{mathtools}

\usepackage{amsmath}
\usepackage{graphicx,epsfig,color}
\usepackage{amsfonts}
\usepackage{subfigure}
\usepackage{bm}
\usepackage{amssymb}
\usepackage{amsthm}
\usepackage{epstopdf}
\usepackage{mathtools}  
\usepackage{tabulary}
\usepackage{booktabs}
\usepackage{cite}

\DeclareSymbolFont{bbold}{U}{bbold}{m}{n}
\DeclareSymbolFontAlphabet{\mathbbold}{bbold}
\newcommand{\vect}[1]{\mathbbold{#1}}
\newcommand{\vectorones}[1][]{\vect{1}_{#1}}
\newcommand{\vectorzeros}[1][]{\vect{0}_{#1}}

\def\diag{\operatorname{diag}}

\newcommand{\subscr}[2]{#1_{\textup{#2}}}

\newtheorem{theorem}{Theorem}
\newtheorem{lemma}[theorem]{Lemma}

\newtheorem{definition}[theorem]{Definition}

\newtheorem{approximation}[theorem]{Approximation}
\newtheorem{assumption}[theorem]{Assumption}
\newtheorem{proposition}[theorem]{Proposition}
\newtheorem{corollary}[theorem]{Corollary}
\newtheorem{notation}[theorem]{Notation}

\title{Competitive Propagation: Models, Asymptotic Behavior and Multi-stage Games 
  \thanks{This material is based upon work supported by, or in part by, the U. S. Army Research Laboratory and the U. S. Army Research Office under grant number W911NF-15-1-0577, and the UCSB Institute for
    Collaborative Biotechnology under grant W911NF-09-D-0001 from the
    U.S.\ Army Research Office.  The content of the information does not
    necessarily reflect the position or the policy of the Government, and
    no official endorsement should be inferred. This paper is related to an
    early conference article~\cite{WM-FB:14a-conference}: this article
    treats a more general and different set of phenomena, with new results on theoretical analysis, complete
    proofs and new simulation results.}}
    
\author{\qquad Wenjun Mei \qquad Francesco Bullo \thanks{Wenjun Mei and
    Francesco Bullo are with the Department of Mechanical Engineering and
    with Center for Control, Dynamical Systems, and Computation, University
    of California, Santa Barbara, Santa Barbara, CA 93106, USA,
    \texttt{meiwenjunbd@gmail.com}, \texttt{bullo@engineering.ucsb.edu}}}

\begin{document}
\maketitle 
\begin{abstract}
  In this paper we propose a class of propagation models for multiple competing products over a social network. We consider two propagation mechanisms: social conversion and self conversion, corresponding, respectively, to endogenous and exogenous factors. A novel concept, the product-conversion graph, is proposed to characterize the interplay among competing products. According to the chronological order of social and self conversions, we develop two Markov-chain models and, based on the independence approximation, we approximate them with two respective difference equations systems. 
  Theoretical analysis on these two approximation models reveals the dependency of the systems' asymptotic behavior on the structures of both the product-conversion graph and the social network, as well as the initial condition. In addition to the theoretical work, accuracy of the independence approximation and the asymptotic behavior of the Markov-chain model are investigated via numerical analysis, for the case where social conversion occurs before self conversion.  
  Finally, we propose a class of multi-player and multi-stage competitive propagation games and discuss the seeding-quality trade-off, as well as the allocation of seeding resources among the individuals. We investigate the unique Nash equilibrium at each stage and analyze the system's behavior when every player is adopting the policy at the Nash equilibrium.
\end{abstract}

 \begin{keywords}
   competitive propagation, independence approximation, network structure, stability analysis, multi-stage uncooperative game, seeding, product quality 
 \end{keywords}

\section{Introduction}

\emph{a) Motivation and problem description}
It is of great scientific interest to model some sociological phenomenon as dynamics on networks, such as consensus, polarization, synchronization and propagation. Indeed, the past fifteen years have witnessed a flourishing of research on propagation of diseases, opinions, commercial products etc, collectively referred to as memes, on social networks. Much progress has been made both statistically~\cite{NAC-JHF:07,NAC-JHF:10,DC:10,LC-YS-ADIK-CM-MF-JHF:14} and theoretically~\cite{HWH:00,MEA-SMM:05,WW-XZ:06,AG-DG-LH-XM-JP:11}. In a more recent set of extensions, scientists have begun studying the simultaneous propagation of multiple memes, in which not only the interaction between nodes (or equivalently referred to as individuals) in the network, but also the interplay of multiple memes, plays an important role in determining the system's dynamical behaviors. These two forms of interactions together add complexity and research value to the multi-meme propagation model.

This paper proposed a series of mathematical models on the propagation of
competing products. Three key elements: the interpersonal network, the individuals and the competing products, are modeled respectively as a graph with fixed topology, the nodes on the graph, and the states of nodes. Our models are based on
the characterization of individuals' decision making behaviors under the
social pressure. Two factors determine individuals' choices on which
product to adopt: the endogenous factor and the exogenous factor. The
endogenous factor is the social contact between nodes via social links,
which forms a tendency of imitation, referred to as social pressure in this
paper. The exogenous factor is what is unrelated to the network, e.g., the
products' quality.

In the microscopic level, we model the endogenous and exogenous factors respectively as two types of product-adoption processes: the social conversion and the self conversion. In social conversion, any node randomly picks one of its neighbors and follows that neighbor's state with some given probability characterizing how open-minded the node is. In the self conversion, each node independently converts from one product to another with some given probability depending on the two products involved. Although individuals exhibit subjective preferences when they are choosing the products, statistics on a large scale of different individuals' behaviors often reveal that the relative qualities of the competing products are objective. For example, although some people may have special affections on feature phones, the fact that more people have converted from feature phones to smart phones, rather than the other way around, indicates that the latter is relatively better. We assume that the transit probabilities between the competing products are determined by their relative qualities and thus homogeneous among the individuals.  

\emph{b) Literature review:} Various models have been proposed to describe the propagation on networks, such as the percolation model on random graphs~\cite{RPS-AV:01,MEJN:02}, the independent cascade model~\cite{JG-BL-EM:01, DK-JK-ET:03, SB-DK-MS:07}, the linear threshold model~\cite{DA-AO-EY:11,EY-DA-AO-AS:11,EA-MD-AO:13} and the epidemic-like mean-field model~\cite{YW-DC-CW-CF:03,PvM-JO-RK:09,MY-CS:11}. The first two are stochastic models while the linear threshold model is usually deterministic. The epidemic-like model is an ODE system as the approximation of a Markov chain, which gives the dynamics of any node's state probability distribution.

As extensions to the propagation of a single meme, some recent papers have
discussed the propagation of multiple memes, e.g.,
see~\cite{LW-AF-AV-FM:12,KA-EM:11, AF-AJ:12,FU-KS:12,BP-AB-RR-CF:12,
  AB-BP-RR-CF:12,FDS-CS:14,AS-DT-LK:14,TC-CN-SMW-AZ:07,SS-BB-HV-AS-PH:12,SG-HH-MK:14,AF-AA-AJ:15}. Some
of these papers adopt a Susceptible-Infected-Susceptible (SIS)
epidemic-like model and discuss the long-term coexistence of multiple memes
in single/multiple-layer networks, e.g.,
see~\cite{BP-AB-RR-CF:12,AB-BP-RR-CF:12,FDS-CS:14}. Some papers focus
instead on the strategy of initial seeding to maximize or prevent the
propagation of one specific meme in the presence of
adversaries~\cite{TC-CN-SMW-AZ:07,SS-BB-HV-AS-PH:12,SG-HH-MK:14,AF-AA-AJ:15}.
Among all these papers mentioned in this paragraph, our model is most
closely related to the work by Stanoev et.\ al.~\cite{AS-DT-LK:14}
but the social contagion process in~\cite{AS-DT-LK:14} is different from our model and theoretical analysis
on the general model is not included

\emph{c) Contribution:} Firstly we propose a generalized and novel model for
the competitive propagation on social networks. By taking into account both
the endogenous and exogenous factors and considering the individual
variance as well as the interplay of the competing products, our model is
general enough to describe a large class of multi-meme propagation
processes. Moreover, in the modeling of multi-meme contagions, many models
come across the problem of dealing with multiple contagions to one node by
different memes at a single time step, which is usually avoided by assuming
the infinitesimal step length so that it only allows for a single contagion
at every step. Differently from these models, the problem of multiple
contagions does not occur in our model since we model the contagion process
as the individual's initiative choice under the social pressure, which is
more suitable for the product-adoption process. In addition, 
compared with the independent cascade model, in which individuals' choices are
irreversible, our models adopt a more realistic assumption that conversions
from one product to another are reversible and occur persistently.

Secondly, we propose a new concept, the product-conversion graph, to characterize the interplay between the products. There are two graphs in our model: the social network describing the interpersonal connections, and the product-conversion graph defining the transitions between the products in self conversion, which in turn reflect the products' relative quality. 

Thirdly, starting from the description of individuals' behavior, we develop two Markov-chain competitive propagation models different in the chronological order of the social conversion and the self conversion processes. Applying the independence approximation, we propose two respective network competitive propagation models, which are difference equations systems, such that the dimension of our problem is reduced and some theorems in the area of dynamical systems can be applied to the analysis of the approximation models.

Fourthly, both theoretical analysis and simulation results are
presented on the dynamical properties of the network competitive
propagation models. We discuss the existence, uniqueness and stability of
the fixed point, as well as how the systems' asymptotic state probability distribution is
determined by the social network structure, the individuals'
open-mindedness, the initial condition and, most importantly, the structure
of the product-conversion graph. We find that, if the product-conversion graph contains only one absorbing strongly connected component, then the self
conversion dominates the system's asymptotic behavior; With multiple absorbing strongly connected components in the product-conversion graph, the system's asymptotic state
probability distribution also depends on the initial condition, the network
topology and the individual open-mindedness.
In addition, simulation results are presented to show the high accuracy of
the independence approximation and reveal that the original Markov-chain
model also exhibits the same asymptotic behavior.

At last, based on the network competitive propagation model, and according to what actions can be taken by the players, we propose two different models of multi-player and multi-stage non-cooperative games, in which the players are the competing companies with bounded investment budgets. Each company can invest on seeding, e.g., advertisement and promotion, or its product's quality. For each game model, we investigate the unique Nash equilibrium at each stage. Theoretical analysis reveals some strategic insights on the seeding-quality trade-off and the allocation of seeding resources. The insights are directive and realistic. We also discuss the systems' behaviors if every player is adopting the policy at the Nash equilibrium. 

\emph{d) Organization:} The rest of this paper is organized as follows. Section II give the assumptions of two Markov-chain propagation models. Section III and IV discuss the approximation of these two models separately. In Section V, we discuss the multi-player and multi-stage competitive propagation games. Section VI is the conclusion.

\section{Model Description and Notations}
\emph{a) Social network as a graph:} In this model, a social network is considered as an undirected, unweighted, fixed-topology graph $G=(V,E)$ with $n$ nodes. The nodes are indexed by $i\in V=\{ 1,2,\dots,n \}$. The adjacency matrix is denoted by $A=(a_{ij})_{n\times n}$ with $a_{ij}=1$ if $(i,j)\in E$ and $a_{ij}=0$ if $(i,j)\notin E$.

The row-normalized adjacency matrix is denoted by $\tilde{A}=(\tilde{a}_{ij})_{n \times n}$, where $\tilde{a}_{ij}=\frac{1}{N_i}a_{ij}$ with $N_i=\sum_{j=1}^n a_{ij}$. The graph $G=(V,E)$ is always assumed connected and there is no self loop, i.e., $\tilde{a}_{ii}=0$ for any $i\in V$.

\emph{b) Competing products and the states of nodes:} Suppose there are $R$ competing products, denoted by $H_1, H_2,\dots,H_R$,  propagating in the network. We consider a discrete-time model, i.e., $t\in \mathbb{N}$, and assume the products are mutually exclusive. We do not specify the state of adopting no product and collectively refer to all the states as ``products''. Denote by $D_i(t)$ the state of node $i$ after time step $t$. For any $t\in \mathbb{N}$, $D_i(t)\in \{ H_1,H_2,\dots,H_R \}$. For simplicity let $\Theta=\{1,2,\dots,R\}$, i.e., the set of the product indexes.

\emph{c) Nodes' production adoption behavior: }Two mechanisms define the individuals' behavior: the social conversion and the self conversion. The following two assumptions propose respectively two models different in the chronological order of the social and self conversions.
\begin{assumption}[Social-self conversion model]\label{asmp_social_self_conv}
Consider the competitive propagation of $R$ products in the network $G=(V,E)$. At time step $t+1$ for any $t\in \mathbb{N}$, suppose the previous state of any node $i$ is $D_i(t)=H_r$. Node $i$ first randomly pick one of its neighbor $j$ and following $j$'s previous state, i.e., $D_i(t+1)=D_j(t)$, with probability $\alpha_i$. If node $i$ does not follow $j$'s state in the social conversion, with probability $1-\alpha_i$, then node $i$ converts to product $H_s$ with probability $\delta_{rs}$ for any $s\neq r$, or stay in $H_r$ with probability $\delta_{rr}$.
\end{assumption}    

\begin{assumption}[Self-social conversion model]\label{asmp_self_social_conv}
At any time step $t+1$, any node $i$ with $D_i(t)=H_r$ converts to $H_s$ with probability $\delta_{rs}$ for any $s\neq r$, or stay in the state $H_r$ with probability $\delta_{rr}$. If node $i$ stays in $H_r$ in the process above, then node $i$ randomly picks a neighbor $j$ and following $D_j(t)$ with probability $\alpha_i$, or still stay in $H_r$ with probability $1-\alpha_i$.
\end{assumption}

\begin{figure}
\vspace{0.08cm}
\begin{center}
\includegraphics[width=.8\linewidth]{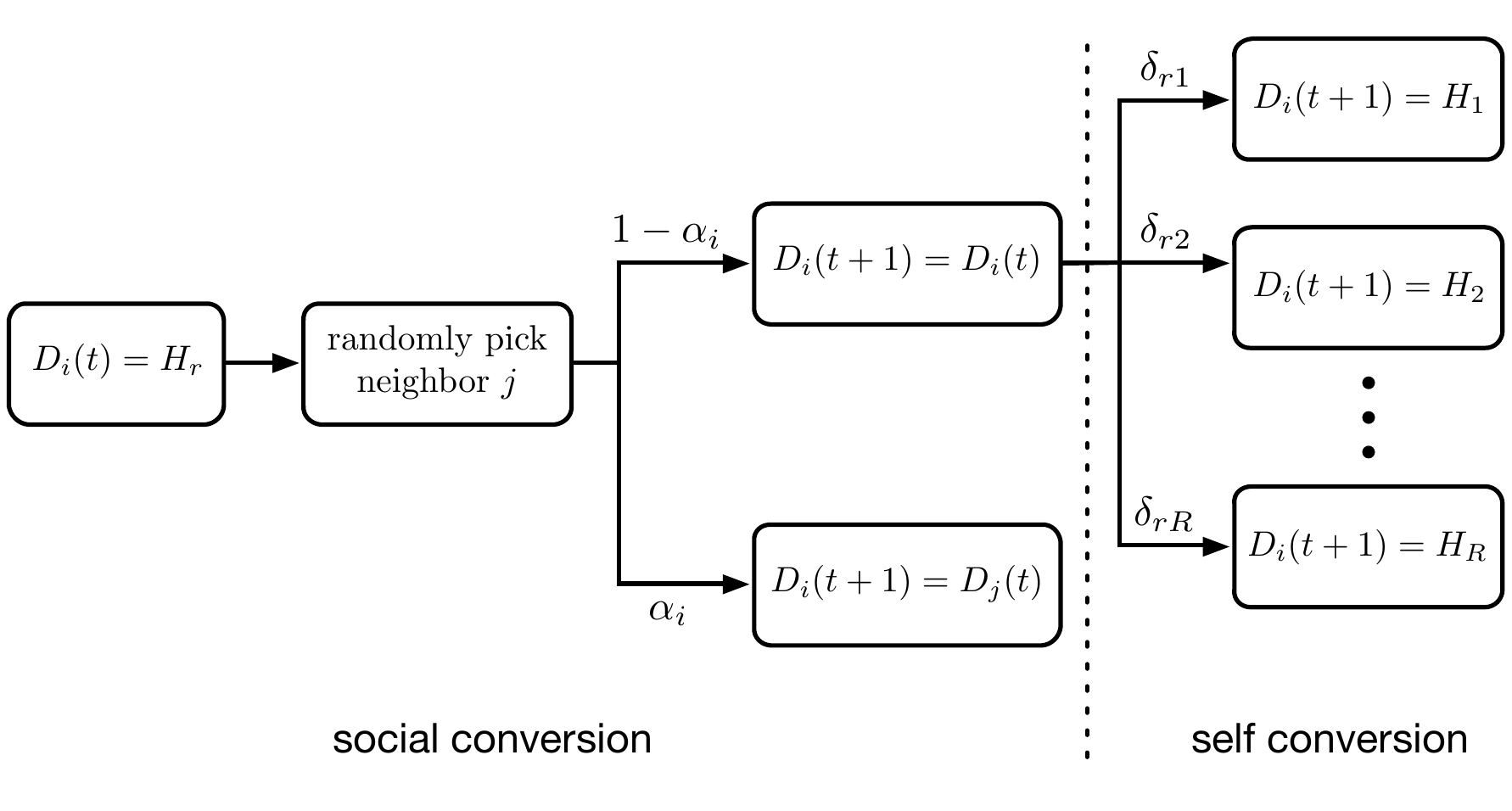}
\caption{Diagram illustration for the social-self conversion model}
\label{fig_diagram_social_self}
\end{center}
\end{figure}

\begin{figure}
\vspace{0.01cm}
\begin{center}
\includegraphics[width=.75\linewidth]{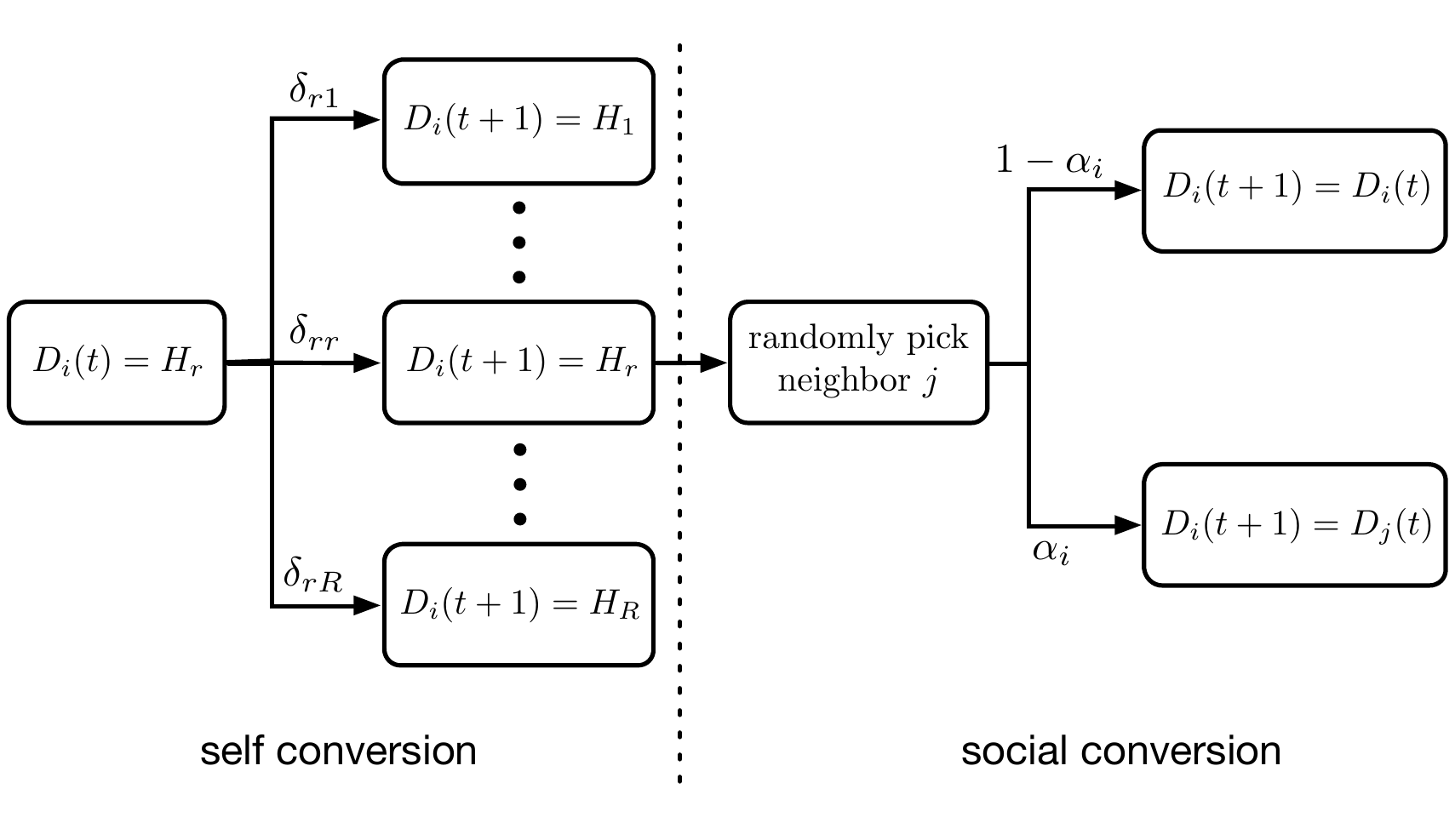}
\caption{Diagram illustration for the self-social conversion model}
\label{fig_diagram_self_social}
\end{center}
\end{figure}

Assumptions~\ref{asmp_social_self_conv} and~\ref{asmp_self_social_conv} are illustrated by Figure~\ref{fig_diagram_social_self} and Figure~\ref{fig_diagram_self_social} respectively. By introducing the parameters $\delta_{rs}$ we define a directed and weighted graph with the adjacency matrix $\Delta=(\delta_{rs})_{R\times R}$, referred to as the \emph{product-conversion graph}. Figure~\ref{fig_prod_conv_graph} gives an example of the product-conversion graph for different smart phone operation systems. Based on either of the two assumptions, $\Delta$ is row-stochastic. In this paper we discuss several types of structures of the product-conversion graph, e.g., the case when it is  strongly connected, or consists of a transient subgraph and some isolated absorbing subgraphs. The parameter $\alpha_i$ characterizes node $i$'s inclination to be influenced by social pressure. Define $\bm{\alpha}=(\alpha_1,\alpha_2,\dots,\alpha_n)^{\top}$ as the individual open-mindedness vector. Assume $0<\alpha_i<1$ for any $i\in V$.  

\begin{figure}
\begin{center}
\includegraphics[width=.65\linewidth]{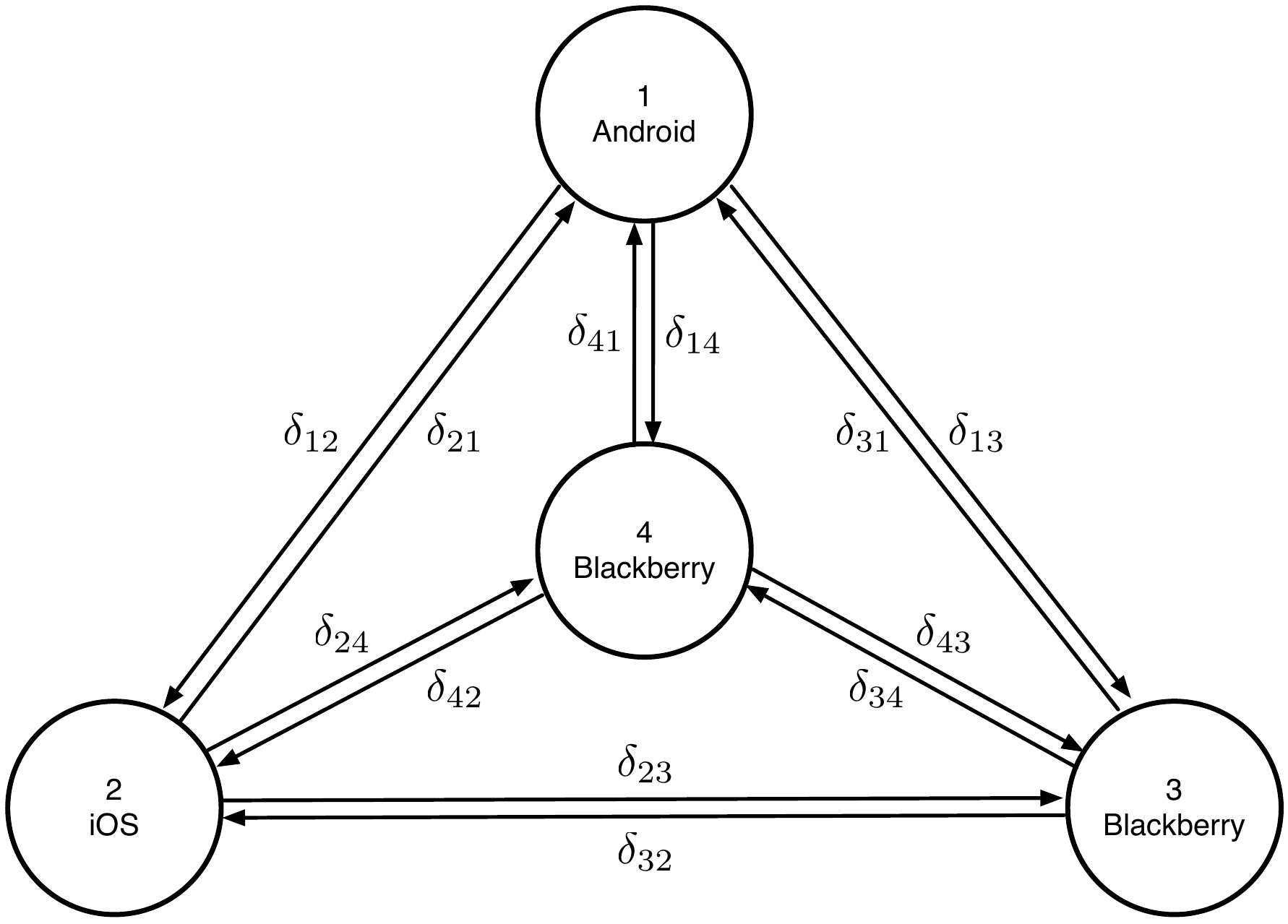}
\caption{An example of the product-conversion graph for different smart phone operation systems.}
\label{fig_prod_conv_graph}
\end{center}
\end{figure}

\emph{d) Problem description:} According to either Assumption~\ref{asmp_social_self_conv} or Assumption~\ref{asmp_self_social_conv}, at any time step $t+1$, the probability distribution of any node's states depends on its own state as well as the states of all its neighbors at time $t$. Therefore, the collective evolution of nodes' states is a $R^n$-state discrete-time Markov chain. Define $p_{ir}(t)$ as the probability that node $i$ is in state $H_r$ after time step $t$, i.e., $p_{ir}(t)=\mathbb{P}[D_i(t)=H_r]$. We aim to understand the dynamics of $p_{ir}(t)$. Since the Markov chain models have exponential dimensions and are difficult to analyze, we approximate it with lower-dimension difference equations systems and analyze instead the dynamical properties of the approximation systems. 

\emph{e) Notations:} Before proceeding to the next section, we introduce some frequently used notations in Table~\ref{table_notations}.
\begin{table}[htbp]\caption{Notations frequently used in this paper}\label{table_notations}
\begin{center}
\begin{tabular}{r p{6cm} }
\toprule
$\succeq$ ($\preceq$) & entry-wise no less(greater) than\\
$\succ$ ($\prec$) & entry-wise strictly greater(less) than\\
$\vectorones[n]$, $\vectorzeros[n]$ & $(1,1,\dots,1)^{\top} \in \mathbb{R}^{n\times 1}$, $(0,0,\dots,0)^{\top}\in \mathbb{R}^{n\times 1}$ \\
$\vectorzeros[n\times m]$ & $(0)_{n \times m}$\\
$S_{nm}(\bm{a})$ & The set $\{\bm{X}\in \mathbb{R}^{n\times m}\,|\, X\succeq \vectorzeros[n\times m],X\vectorones[m]=\bm{a}\}$ for any $\bm{a}\in \mathbb{R}^n$\\
$\tilde{S}_{nm}(\bm{a})$ & The set $\{\bm{X}\in \mathbb{R}^{n\times m}\,|\, X\succeq \vectorzeros[n\times m],X\vectorones[m]\preceq \bm{a}\}$ for any $\bm{a}\in \mathbb{R}^n$\\
$\bm{w}(M)$ & The normalized dominant left eigenvector for matrix $M\in \mathbb{R}^{l\times l}$ if it has one\\
$\bm{x}^r$ & The $r$-th column vector of the matrix $X\in \mathbb{R}^{n\times m}$\\
$\bm{x}^{(i)}$ & The $i$-th row vector of the matrix $X\in \mathbb{R}^{n\times m}$\\
$\bm{x}^{(-i)}$ & The $i$-th row vector of the matrix $\tilde{A}X\in \mathbb{R}^{n\times m}$, i.e., $\bm{x}^{(-i)}=(x_{-i1},x_{-i2},\dots,x_{-im})$ where $x_{-ir}=\sum_{j=1}^n \tilde{a}_{ij}x_{jr}$\\
$G(A)$ & The graph with the adjacency matrix $A$\\
\bottomrule
\end{tabular}
\end{center}
\end{table} 




\section{Network Competitive Propagation Model with Social-self conversion}
This section is based on Assumption~\ref{asmp_social_self_conv}. We first derive an approximation model for the time evolution of $p_{ir}(t)$, referred to as the \emph{social-self conversion network competitive propagation model} (social-self NCPM), and then analyze the asymptotic behavior of the approximation model and its relation to the social network topology, the product-conversion graph, the initial condition and the individuals open-mindedness. Further simulation work is presented in the end of this section.

\subsection{Derivation of the social-self NCPM}
Some notations are used in this section.
\begin{notation}
For the competitive propagation of products $\{ H_1,H_2,\dots,H_R \}$ on the network $G=(V,E)$, 

(1) define the random variable $X_i^r(t)$ by
\begin{equation*}
X_i^r(t)=
\begin{cases}
\displaystyle 1, & \quad \text{if }D_i(t)=H_r,\\
\displaystyle 0, & \quad \text{if }D_i(t)\neq H_r.
\end{cases}
\end{equation*}
According to the mutual exclusiveness of the competing products, for any $i\in V$, if $X_i^r(t)=1$, then $X_i^s(t)=0$ for any $s\neq r$;

(2) Define the $n-1$ tuple $\bm{D}_{-i}(t)=(D_1(t),\dots, D_{i-1}(t),D_{i+1}(t),\dots, D_n(t))$, i.e., the states of all the nodes except node $i$ after time step $t$;

(3) Define the following notations for simplicity:
\begin{align*}
P_{ij}^{rs}(t) & =\mathbb{P}[X_i^r(t)=1\,|\, X_j^s(t)=1],\\
P_i^r(t;-i) & =\mathbb{P}[X_i^r(t)=1\,|\, \bm{D}_{-i}(t)],\\
\Gamma_i^r(t;s,-i) & =\mathbb{P}[X_i^r(t+1)=1\,|\, X_i^s(t)=1,\bm{D}_{-i}(t)].
\end{align*}
\end{notation}

In the derivation of the network competitive propagation model, the following approximation is adopted:
\begin{approximation}[Independence Approximation]\label{indep_approx}
For the competitive propagation of $R$ products on the network $G=(V,E)$, when deriving the equation for $p_{ir}(t)$, approximate the conditional probability $P_{ij}^{ms}(t)$ by its corresponding total probability $p_{im}(t)$ for any $m,s\in \Theta$ and any $i,j\in V$.
\end{approximation}

With the \emph{independence approximation}, the social-self NCPM is presented in the theorem below.
\begin{theorem}[Social-self NCPM]\label{thm_social_self_mf_model}
Consider the competitive propagation based on Assumption~\ref{asmp_social_self_conv}, with the social network and the product-conversion graph represented by their adjacency matrices $\tilde{A}=(\tilde{a}_{ij})_{n\times n}$ and $\Delta=(\delta_{rs})_{R\times R}$ respectively. The probability $p_{ir}(t)$ satisfies
\begin{equation}\label{eq_social_self_NCPM_exact}
\begin{split}
p_{ir}&(t+1)-p_{ir}(t) \\
& = \sum_{s\neq r} \alpha_i \sum_{j=1}^n \tilde{a}_{ij} \big( P_{ij}^{sr}(t)p_{jr}(t)-P_{ij}^{rs}(t)p_{js}(t) \big)\\
& \quad + \sum_{s\neq r}(1-\alpha_i)(\delta_{sr}p_{is}(t)-\delta_{rs}p_{ir}(t)),
\end{split}
\end{equation}
for any $i\in V$ and $r\in \Theta$. Applying the independence approximation, the approximation model for equation~\eqref{eq_social_self_NCPM_exact}, i.e., the social-self NCPM, is
\begin{equation}\label{eq_mf_social_self_NCPM_entry}
\begin{split}
p_{ir}( & t+1)\\
&=\alpha_i \sum_{j=1}^n \tilde{a}_{ij}p_{jr}(t) + (1-\alpha_i)\sum_{s=1}^R \delta_{sr}p_{is}(t).
\end{split}
\end{equation} 
\end{theorem} 
\smallskip
\begin{proof}
By definition,
\begin{equation*}
\begin{split}
p_{ir}(t+ & 1)-p_{ir}(t)\\
& =\mathbb{E}[X_i^r(t+1)-X_i^r(t)]\\
& = \mathbb{E}\big[ \mathbb{E}[X_i^r(t+1)-X_i^r(t)\,|\, \bm{D}_{-i}(t)] \big],
\end{split}
\end{equation*}
where the conditional expectation is given by
\begin{equation*}
\begin{split}
\mathbb{E}[X_i^r(t+&1)-X_i^r(t)\,|\, \bm{D}_{-i}(t)]\\
& =\mathbb{P}[X_i^r(t+1)-X_i^r(t)=1\,|\,\bm{D}_{-i}(t)]\\
&\quad -\mathbb{P}[X_i^r(t+1)-X_i^r(t)=-1\,|\, \bm{D}_{-i}(t)]\\
& =\sum_{s\neq r} \Gamma_i^r(t;s,-i)P_i^s(t;-i)\\
&\quad -\sum_{s\neq r}\Gamma_i^s(t;r,-i)P_i^r(t;-i).
\end{split}
\end{equation*}
According to Assumption~\ref{asmp_social_self_conv},
\begin{equation*}
\begin{split}
&\Gamma_i^r(t;s,-i)P_i^s(t;-i)\\
&\text{ }=\alpha_i \sum_j \tilde{a}_{ij}X_j^r(t)P_i^s(t;-i)+(1-\alpha_i)\delta_{sr}P_i^s(t;-i).
\end{split}
\end{equation*}
Therefore, 
\begin{equation*}
\begin{split}
\mathbb{E}[\Gamma_i^r(t;s,-i)P_i^s(t;-i)] & =\alpha_i \sum_j \tilde{a}_{ij} \mathbb{E}[X_j^r(t)P_i^s(t;-i)]\\
&\quad +(1-\alpha_i)\delta_{sr}\mathbb{E}[P_i^s(t;-i)].
\end{split}
\end{equation*}
One the right-hand side of the equation above, $\mathbb{E}[P_i^s(t;-i)]=p_{is}(t)$. Moreover,
\begin{equation*}
\begin{split}
&\mathbb{E}[X_j^r(t)P_i^s(t;-i)]\\
&\text{ } =\sum_{\bm{d}_{-i-j}}\mathbb{P}[X_i^s(t)=1\,|\, \bm{D}_{-i-j}(t), X_j^r(t)=1]\\
&\text{ }\quad \cdot \mathbb{P}[\bm{D}_{-i-j}(t)=\bm{d}_{-i-j},X_j^r(t)=1]\\
&\text{ } =\sum_{\bm{d}_{-i-j}} \mathbb{P}[X_i^s(t)=1,X_j^r(t)=1,\bm{D}_{-i-j}(t)=\bm{d}_{-i-j}]\\
&\text{ } =P_{ij}^{sr}(t)p_{jr}(t).
\end{split}
\end{equation*}
Apply the same computation to $\mathbb{E}[\Gamma_i^s(t;r,-i)P_i^r(t;-i)]$ and then we obtain equation~\eqref{eq_social_self_NCPM_exact}. Replace $P_{ij}^{sr}(t)$ and $P_{ij}^{rs}(t)$ by $p_{is}(t)$ and $p_{ir}(t)$ respectively and according to the equations $\sum_{s\neq r}p_{is}(t)=1-p_{ir}(t)$ and $\sum_{s\neq r}\delta_{rs}=1-\delta_{rr}$, we obtain equation~\eqref{eq_mf_social_self_NCPM_entry}.
\end{proof}

The derivation of Theorem~\ref{thm_social_self_mf_model} is equivalent to the widely adopted mean-field approximation in the modeling of the network epidemic spreading~\cite{PvM-JO-RK:09,MB-JM:12,FDS-CS-PvM:13}. 
Notice that the independence approximation neither neglects the correlation between any two nodes' states, nor destroys the network topology, since $p_{jr}(t)$, $p_{js}(t)$ and $\tilde{a}_{ij}$ all appear in the dynamics of $p_{ir}(t)$.   

\subsection{Asymptotic behavior of the social-self NCPM}
Define the map $f: \mathbb{R}^{n\times R}\rightarrow \mathbb{R}^{n\times R}$ by
\begin{equation}\label{eq_social_self_matrix_f_map}
f(X)=\diag(\bm{\alpha})\tilde{A}X+(I-\diag(\bm{\alpha}))X\Delta.
\end{equation} 
According to equation~\eqref{eq_mf_social_self_NCPM_entry}, the matrix form of the social-self NCPM is written as
\begin{equation}\label{eq_mf_social_self_matrix}
P(t+1)=f\big( P(t) \big),
\end{equation}
where $P(t)=(p_{ir}(t))_{n\times R}$. We analyze how the asymptotic behavior of system~\eqref{eq_mf_social_self_matrix}, i.e., the existence, uniqueness and stability of the fixed point of the map $f$, is determined by the two graphs introduced in our model: the social network with the adjacency matrix $\tilde{A}$, and the product-conversion graph with the adjacency matrix $\Delta$.


\subsubsection{Structures of the social network and the product-conversion graph}
Assume that the social network $G(\tilde{A})$ has a globally reachable node. As for the product-conversion graph, we consider the more general case. Suppose that the product-conversion graph $G(\Delta)$ has $m$ absorbing strongly connected components (absorbing SCCs) and a transient subgraph. Re-index the products such that the product index set for any $l$-th absorbing SCCs is given by 
\begin{align*}
\Theta_1 & = \{1,2,\dots,k_1\},\text{ and}\\
\Theta_l & = \{ \sum_{u=1}^{l-1} k_u+1,\sum_{u=1}^{l-1}k_u+2,\dots,\sum_{u=1}^l k_u \},
\end{align*}
for any $l\in \{2,3,\dots,m\}$, and the index set for the transient subgraph is $\Lambda=\{\sum_{l=1}^m k_l+1,\dots, \sum_{l=1}^m k_l+2,\dots, R\}$. then the adjacency matrix $\Delta$ of the product-conversion graph takes the following form:
\begin{equation}\label{eq_Delta_combined_struc}
\Delta=
\begin{bmatrix}
\bar{\Delta} & \vectorzeros[(R-k_0)\times k_0]\\
 B_{k_0\times (R-k_0)} & \Delta_0
\end{bmatrix},
\end{equation}
where $\bar{\Delta}=\diag[\Delta_1,\Delta_2,\dots,\Delta_m]$ and $B=[B_1,B_2,\dots,B_m]$, with $B_l\in \mathbb{R}^{k_0\times k_1}$ for any $l\in \{1,2,\dots,m\}$, is nonzero and entry-wise non-negative.
Matrix $\Delta_l=(\delta_{rs}^{\Theta_l})_{k_l\times k_l}$, with $\delta_{rs}^{\Theta_1}=\delta_{rs}$ and $\delta_{rs}^{\Theta_l}=\delta_{\sum_{u=1}^{l-1} k_u+r,\sum_{u=1}^{l-1} k_u+s}$ for any $l \in \{2,3,\dots,m\}$, is the adjacency matrix of the $l$-th absorbing SCC, and is thus irreducible and row-stochastic. The following definition classifies four types of structures of $G(\Delta)$.

\begin{definition}[Four sets of product-conversion graphs]\label{def_foure_cases_product_conv_graph}
  Based on whether the product-conversion graph $G(\Delta)$ has a
  transient subgraph and a single or multiple absorbing SCCs, we
  classify the adjacency matrix $\Delta$ into the following four
  cases:
  \begin{enumerate}
\item Case 1 (single SCC): The graph $G(\Delta)$ is strongly connected, i.e., $\Delta=\Delta_1$, with $k_1=R$;
\item Case 2 (single SCC + transient subgraph): The graph $G(\Delta)$ contains one absorbing SCC and a transient subgraph, i.e., $\bar{\Delta}=\Delta_1$ and $k_0\ge 1$; 
\item Case 3 (multi-SCC): The graph $G(\Delta)$ contains $m$ absorbing SCCs, i.e., $\Delta=\diag[\Delta_1,\Delta_2,\dots,\Delta_m]$, with $\sum_{l=1}^m k_l=R$;
\item Case 4 (multi-SCC + transient subgraph): The graph $G(\Delta)$ contains $m$ absorbing SCCs and a transient subgraph, with $\Delta$ given by equation~\eqref{eq_Delta_combined_struc}.
\end{enumerate}
\end{definition}

\subsubsection{Stability analysis of the social-self NCPM}
The following theorem states the distinct asymptotic behaviors of the
social-self NCPM, with different structures of the product-conversion
graph.
\begin{theorem}[Asymptotic behavior for social-self NCPM]\label{thm_asym_behav_social_self_NCPM}
Consider the social-self NCPM on a strongly connected social network $G(\tilde{A})$, with the product-conversion graph $G(\Delta)$. Assume that
\begin{enumerate}
\item Each absorbing SCC $G(\Delta_l)$ of $G(\Delta)$ is aperiodic;
\item For any $\Delta_l$, $l\in \{1,2,\dots,m\}$, as least one column of $\Delta_l$ is entry-wise strictly positive;
\item For any $r\in \Lambda$, $\sum_{s\in \Lambda}\delta_{rs}<1$, i.e., $\Delta_0 \vectorones[k_0]\prec \vectorones[k_0]$.
\end{enumerate}
Then, for any $P(0)\in S_{nR}(\vectorones[n])$, the solution $P(t)$ to
equation~\eqref{eq_mf_social_self_matrix} has the following
properties, depending upon the structure of $\Delta$:
\begin{enumerate}
\item in Case 1, $P(t)$ converges to $P^*=\vectorones[n]\bm{\Delta}^{\top}$ exponentially fast, where $P^*$ is the unique fixed point in $S_{nR}(\vectorones[n])$ for the map $f$ defined by equation~\eqref{eq_social_self_matrix_f_map}. Moreover, the convergence rate is $\epsilon(\Delta)=\subscr{\alpha}{max}+(1-\subscr{\alpha}{max})\zeta(\Delta)$, where $\subscr{\alpha}{max}=\max_i\alpha_i$ and $\zeta(\Delta)=1-\sum_{r=1}^R \min_s \delta_{sr}$;
\item in Case 2, for any $i\in V$, 
\begin{equation*}
\lim_{t\rightarrow \infty} p_{ir}(t)=
\begin{cases}
\displaystyle 0,&\quad \text{for any }r\in \Lambda,\\
\displaystyle w_r(\Delta_1),& \quad \text{for any }r\in \Theta_1;
\end{cases}
\end{equation*}
\item in Case 3, for any $l\in \{1,2,\dots,m\}$ and $i\in V$,
\begin{equation*}
\lim_{t\rightarrow \infty}\bm{p}^{\Theta_l(i)}(t)=\big( \bm{w}^{\top}(M)P^{\Theta_l}(0)\vectorones[k_l] \big)\bm{w}^{\top}(\Delta_l),
\end{equation*}
where $M=\diag(\bm{\alpha})\tilde{A}+I-\diag{\bm{\alpha}}$ and $P^{\Theta_l}(t)=\big( p_{ir}^{\Theta_l}(t) \big)_{n\times k_l}$, with $p_{ir}^{\Theta_l}(t)=p_{i,\sum_{u=1}^{l-1}k_u+r}(t)$ and $\bm{p}^{\Theta_l(i)}(t)$ being the $i$-th row of $P^{\Theta_l}(t)$;
\item in Case 4, for any $l\in \{1,2,\dots,m\}$ and $i\in V$,
\begin{equation*}
\lim_{t\rightarrow \infty} p_{ir}(t)=
\begin{cases}
\displaystyle 0,&\quad \text{for any }r\in \Lambda,\\
\displaystyle \gamma_l w_r(\Delta_l),& \quad \text{for any }r\in \Theta_l,
\end{cases}
\end{equation*}
where $\gamma_l$ depends on $\tilde{A}$, $B_l$, $P^{\Theta_l}(0)$, $P^{\Lambda}(0)$ and satisfies $\sum_{l=1}^m \gamma_l=1$.
\end{enumerate}
\end{theorem} 
\smallskip

Before proving the theorem above, a useful and well-known lemma is stated without the proof.
\begin{lemma}[Row-stochastic  matrices after pairwise-difference similarity transform]\label{lem_pairwise_diff_transform}
Let $M\in \mathbb{R}^{n\times n}$ be row-stochastic. Suppose the graph $G(M)$ is aperiodic and has a globally reachable node. Then the nonsingular matrix 
\begin{displaymath}
    Q=
    \begin{bmatrix}
      -1      &1 &      &     \\
      &\ddots  &\ddots      & \\           
      &      &-1&      1 \\ 
      1/n&      \dots &1/n    &      1/n
    \end{bmatrix}
\end{displaymath}
satisfies
\begin{equation*}
QMQ^{-1}=
\begin{bmatrix}
\subscr{M}{red} & \vectorzeros[n-1]\\
\bm{c}^{\top} & 1
\end{bmatrix}
\end{equation*}
for some $\bm{c}\in \mathbb{R}^{n-1}$ and $\subscr{M}{red}\in \mathbb{R}^{(n-1)\times (n-1)}$. Moreover, $\subscr{M}{red}$ is discrete-time exponentially stable. 
\end{lemma}
\smallskip

\textbf{Proof of Theorem~\ref{thm_asym_behav_social_self_NCPM}:} (1) Case 1:

Since matrix $\Delta$ is row-stochastic, irreducible and aperiodic, according to the Perron-Frobenius theorem, $\bm{w}(\Delta)\in \mathbb{R}^R$ is
well-defined. By substituting $P^*$, defined by $\bm{p}^{*(i)}=\bm{w}(\Delta)^{\top}$
for any $i\in V$, into equation~\eqref{eq_social_self_matrix_f_map}, we
verify that $P^*$ is a fixed point of $f$. 

For any $X$ and $Y\in \mathbb{R}^{n\times R}$, define the distance $d(\cdot, \cdot)$ by $d(X,Y)=\lVert X-Y \rVert_{\infty}$. Then $(S_{nR}(\vectorones[n]),d)$ is a complete metric space. For any $X\in S_{nR}(\vectorones[n])$, it is easy to check that $f(X)\succeq \vectorzeros[n\times R]$ and
\begin{equation*}
f(X)\vectorones[R]=\diag(\bm{\alpha})\tilde{A}X\vectorones[R]+(I-\diag(\bm{\alpha}))X\vectorones[R]= \vectorones[n].
\end{equation*}
Therefore, $f$ maps $S_{nR}(\vectorones[n])$ to $S_{nR}(\vectorones[n])$.

For any $X\in S_{nR}(\vectorones[n])$, according to equation~\eqref{eq_social_self_matrix_f_map},
\begin{equation}\label{eq_soc_sel_conencted_row_difference_1}
\begin{split}
\lVert f(X)^{(i)}-&f(P^*)^{(i)} \rVert_1 \\
& \le\alpha_i \lVert \bm{x}^{(-i)}-\bm{p}^{*(-i)} \rVert_1\\
&\quad +(1-\alpha_i)\lVert (\bm{x}^{(i)}-\bm{p}^{*(i)})\Delta \rVert_1.
\end{split}
\end{equation}
The first term of the right-hand side of~\eqref{eq_soc_sel_conencted_row_difference_1} satisfies
\begin{equation*}
\begin{split}
\lVert \bm{x}^{(-i)}-\bm{p}^{*(-i)} \rVert_1 & =\sum_{r=1}^R \lvert \sum_{j=1}^n \tilde{a}_{ij}x_{jr}-w_r(\Delta) \rvert\\
& \le \sum_{r=1}^R \sum_{j=1}^n \tilde{a}_{ij}\lvert x_{jr}-w_r(\Delta) \rvert\\
& \le \lVert X-P^* \rVert_{\infty}.
\end{split}
\end{equation*}
The second term of the right-hand side of~\eqref{eq_soc_sel_conencted_row_difference_1} satisfies
\begin{equation*}
\lVert (\bm{x}^{(i)}-\bm{p}^{*(i)})\Delta \rVert_1 = \sum_{r=1}^R \lvert \sum_{s=1}^R\big(x_{is}-w_s(\Delta)\big)\delta_{sr} \rvert.
\end{equation*}
If $\bm{x}^{(i)}=\bm{p}^{*(i)}$, then 
\begin{equation*}
\lVert f(X)^{(i)}-f(P^*)^{(i)} \rVert_1 \le \alpha_i \lVert X-P^* \rVert_{\infty}.
\end{equation*}
If $\bm{x}^{(i)}\neq \bm{p}^{*(i)}$, since
$\bm{x}^{(i)}\vectorones[R]=\bm{p}^{*(i)}\vectorones[R]=1$, both the set $\theta_1=\{ s\,|\,x_{is}\ge w_s(\Delta) \}$ and the set $\theta_2=\{ s\,|\,x_{is}<w_s(\Delta) \}$ are nonempty and  
\begin{equation*}
\begin{split}
\sum_{s\in \theta_1} \big( x_{is}-w_s(\Delta) \big)&=\sum_{s\in \theta_2} \big( w_s(\Delta)-x_{is} \big)\\
&=\frac{1}{2}\sum_{s=1}^R \lvert x_{is}-w_s(\Delta) \rvert.
\end{split}
\end{equation*}
Therefore, 
\begin{equation}\label{eq_soc_sel_connected_estimate_row_diff}
\begin{split}
\lVert (\bm{x}^{(i)}- & \bm{p}^{*(i)})\Delta \rVert_1\\
& = \sum_{r=1}^R \sum_{s=1}^R \lvert x_{is}-w_s(\Delta) \rvert \delta_{sr}\\
& \quad -2\sum_{r=1}^R \min \{ \sum_{s\in \theta_1}(x_{is}-w_s(\Delta))\delta_{sr},\\
& \qquad \qquad \qquad \text{ } \sum_{s\in \theta_2}(w_s(\Delta)-x_{is})\delta_{sr} \},
\end{split}
\end{equation}
where
\begin{equation*}
\begin{split}
\min \{ \sum_{s\in \theta_1}( & x_{is}-w_s(\Delta))\delta_{sr}, \sum_{s\in \theta_2}(w_s(\Delta)-x_{is})\delta_{sr} \}\\
& \ge \frac{1}{2} \min_s \delta_{sr} \lVert \bm{x}^{(i)}-\bm{p}^{*(i)} \rVert_1.
\end{split}
\end{equation*}
Substituting the inequality above into~\eqref{eq_soc_sel_connected_estimate_row_diff}, we obtain
\begin{equation*}
  \lVert (\bm{x}^{(i)}-\bm{p}^{*(i)})\Delta \rVert_1 \le
  \Big(1-\sum_{r=1}^R \min_s \delta_{sr}\Big) \lVert \bm{x}^{(i)}-\bm{p}^{*(i)} \rVert_1. 
\end{equation*}
Since $\sum_{r=1}^R \delta_{sr}=1$ for any $s$, $\sum_{r=1}^R \min_s \delta_{sr}$ is no larger than $1$. In addition, since at least one column of $\Delta$ is strictly positive, $\sum_{r=1}^R \min_s \delta_{sr}>0$. Therefore, 
\begin{equation*}
0\le \zeta(\Delta)=1-\sum_{r=1}^R \min_s \delta_{sr}<1,\text{ and }
\end{equation*}
\begin{equation*}
\lVert f(X)^{(i)}-\bm{p}^{*(i)} \rVert_1 \le \big( \alpha_i+(1-\alpha_i)\zeta(\Delta) \big) \lvert X-P^* \rVert_{\infty}.
\end{equation*}
This leads to 
\begin{equation*}
\lVert f(X)-f(P^*) \rVert_{\infty} \le \epsilon(\Delta)\lVert X-P^* \rVert_{\infty},
\end{equation*}
for any $X\in S_{nR}(\vectorones[n])$ and $0<\epsilon(\Delta)<1$. This concludes the proof for Case 1.

$\text{ }$

(2) Case 2:

For the transient subset $\Lambda$, define $P^{\Lambda}(t)=\big( p^{\Lambda}_{ir}(t) \big)_{n\times k_0}$, with $p^{\Lambda}_{ir}(t)=p_{i,r+k_1}(t)$, for any $i\in V$ and $r\in \{ 1,2,\dots,k_0 \}$. Then,
\begin{equation*}
P^{\Lambda}(t+1)=\diag(\bm{\alpha})\tilde{A}P^{\Lambda}(t)+(I-\diag(\bm{\alpha}))P^{\Lambda}(t)\Delta_0.
\end{equation*}
According to Assumption (iii) of Theorem~\ref{thm_asym_behav_social_self_NCPM}, 
\begin{equation*}
c=\max_{r\in \{ 1,2,\dots,k_0 \}} \sum_{s=1}^{k_0}\delta^{\Lambda}_{rs}<1,\quad \text{and}\quad\Delta_0 \vectorones[k_0]\le c\vectorones[k_0].
\end{equation*}
Therefore,
\begin{equation*}
\begin{split}
P^{\Lambda}(t+&1)\vectorones[k_0]\\
&\preceq \Big( \diag(\bm{\alpha})\tilde{A}+c\big( I-\diag(\bm{\alpha}) \big) \Big)P^{\Lambda}(t)\vectorones[k_0].
\end{split}
\end{equation*}
Since $\rho\Big( \diag(\bm{\alpha})\tilde{A}+c\big( I-\diag(\bm{\alpha}) \big) \Big)<1$, for any $P^{\Lambda}(0)\in \tilde{S}_{nk_0}(\vectorones[n])$, $P^{\Lambda}(t)\rightarrow \vectorzeros[n\times k_0]$ exponentially fast.

Define $P^{\Theta_1}(t)=(p_{ir}(t))_{n\times k_1}$. then we have
\begin{equation*}
\begin{split}
P^{\Theta_1}&(t+1)\\
& = \diag(\bm{\alpha})\tilde{A}P^{\Theta_1}(t)+\big( I-\diag(\bm{\alpha}) \big)P^{\Theta_1}(t)\Delta_1\\
& \quad + \big( I-\diag(\bm{\alpha}) \big)P^{\Lambda}(t)B.
\end{split}
\end{equation*}
Since $P^{\Lambda}(t)$ converges to $\vectorzeros[n\times k_0]$ exponentially fast, we have: 1) there exists $C>0$ and $0<\xi <1$ such that
\begin{equation*}
\lVert \big( I-\diag(\bm{\alpha})P^{\Lambda}(t)B \big) \rVert_{\infty} \le C\xi^t;
\end{equation*}
2) $\lVert P^{\Theta_1}(t)\vectorones[k_1]-\vectorones[k_1] \rVert_{\infty}\rightarrow 0$ exponentially fast, which implies $d\big(P^{\Theta_1}(t),S_{nk_1}(\vectorones[n]) \big)\rightarrow 0$ exponentially fast.

For any $X\in \tilde{S}_{nk_1}(\vectorones[n])$, define map $\tilde{f}$ by 
\begin{equation*}
\tilde{f}(X) = \diag(\bm{\alpha})\tilde{A}X+\big( I-\diag(\bm{\alpha}) \big)X\Delta_1.
\end{equation*}
According to the proof for Case 1, there exists a unique fixed point $\tilde{P}^*$ for the map $\tilde{f}$ in $S_{nk_1}(\vectorones[n])$,  given by $\tilde{p}_{ir}^*=w_r(\Delta_1)$. Moreover, there exists $0<\epsilon<1$ such that, for any $X\in S_{nk_1}(\vectorones[n])$,
\begin{equation*}
\lVert \tilde{f}(X)-\tilde{P}^* \rVert_{\infty} \le \epsilon \lVert X-\tilde{P}^* \rVert_{\infty}.
\end{equation*}
Since the function $\frac{\lVert \tilde{f}(X)-\tilde{P}^* \rVert_{\infty}}{\lVert X-\tilde{P}^* \rVert_{\infty}}$ is continuous in $\tilde{S}_{nk_1}(\vectorones[n])\setminus \tilde{P}^*$ and $d\big(P^{\Theta_1}(t),S_{nk_1}(\vectorones[n]) \big)\rightarrow 0$, there exists $T>0$ and $0<\eta<1$ such that, for any $t>T$,
\begin{equation*}
\lVert \tilde{f}\big( P^{\Theta_1}(t)\big)-\tilde{P}^*  \rVert_{\infty} \le \eta \lVert P^{\Theta_1}(t)-\tilde{P}^* \rVert_{\infty}.
\end{equation*}

For $t\in \mathbb{N}$ much larger than $T$,
\begin{equation*}
\begin{split}
\lVert P^{\Theta_1}&(t)-\tilde{P}^* \rVert_{\infty}\\
&\le \eta \lVert P^{\Theta_1}(t-1)-\tilde{P}^* \rVert_{\infty} + C\xi^t\\
& = \eta^{t-T} \lVert P^{\Theta_1}(T)-\tilde{P}^* \rVert_{\infty} + C\frac{\xi^{t}-\eta^{t-T}\xi^{T}}{\eta/\xi}.
\end{split}
\end{equation*}
Since $0<\eta<1$, $0<\xi<1$ as $t\rightarrow \infty$, $\lVert P^{\Theta_1}(t)-\tilde{P}^* \rVert_{\infty}\rightarrow 0$. This concludes the proof for Case 2.

$\text{ }$

(3) Case 3:

For any $l\in \{1,2,\dots,m\}$, 
\begin{equation*}
\begin{split}
P^{\Theta_l}(&t+1)\\
& =\hat{f}\big( P^{\Theta_l}(t) \big)\\
& = \big( I-\diag(\bm{\alpha}) \big)P^{\Theta_l}(t)\Delta_l+\diag(\bm{\alpha})\tilde{A}P^{\Theta_l}(t),
\end{split}
\end{equation*}
where $\Delta_l \vectorones[k_l]=\vectorones[k_l]$ since $\Theta_l$ is absorbing and strongly connected. Therefore,
\begin{equation*}
P^{\Theta_l}(t+1)\vectorones[k_l]=MP^{\Theta_l}(t)\vectorones[k_l],
\end{equation*}
where $M=I-\diag(\bm{\alpha})+\diag(\bm{\alpha})\tilde{A}$ is row-stochastic and aperiodic. Moreover, the graph $G(M)$ has a globally reachable node and therefore the matrix $M$ has a normalized dominant left eigenvector $\bm{w}(M)$. Applying the Perron-Frobenius theorem,
\begin{equation*}
\lim_{t\rightarrow \infty} P^{\Theta_l}(t)\vectorones[k_l] = \big( \bm{w}^{\top}(M)P^{\Theta_l}(0)\vectorones[k_l] \big)\vectorones[n].
\end{equation*}

Let $c_l=\bm{w}^{\top}(M)P^{\Theta_l}(0)\vectorones[k_l]$. Following the same line of argument in the proof for Case 2, $\hat{f}$ maps $S_{nk_l}(c_l \vectorones[n])$ to $S_{nk_l}(c_l \vectorones[n])$, and maps $\tilde{S}_{nk_l}(c_l \vectorones[n])$ to $\tilde{S}_{nk_l}(c_l \vectorones[n])$. Moreover, $\hat{P}^*\in \mathbb{R}^{n\times k_l}$ with $\hat{\bm{p}}^{*(i)}=c_l \bm{w}^{\top}(\Delta_l)$, for any $i\in V$, is the unique fixed point of the map $\hat{f}$ in $S_{nk_l}(c_l\vectorones[n])$. In addition, there exists $0<\epsilon<1$ such that for any $X\in S_{nk_l}(c_l\vectorones[n])$,
\begin{equation*}
\lVert \hat{f}(X)-\hat{P}^* \rVert_{\infty} \le \epsilon \lVert X-\hat{P}^* \rVert_{\infty}.
\end{equation*}

The function $\hat{h}(X)=\frac{\lVert \hat{f}(X)-\hat{P}^* \rVert_{\infty}}{\lVert X-\hat{P}^* \rVert_{\infty}}$ is continuous in $\tilde{S}_{nk_l}(c_l \vectorones[n])\setminus \hat{P}^*$. Since for any $P^{\Theta_l}(0)\in \tilde{S}_{nk_l}(c_l \vectorones[n])\setminus \hat{P}^*$, we have $P^{\Theta_l}(t)\vectorones[k_l]\rightarrow c_l\vectorones[k_l]$, which implies $d\big(P^{\Theta_l}(t),S_{nk_l}(c_l \vectorones[k_l])\big)\rightarrow 0$ as $t\rightarrow 0$. Therefore, there exists $0<\eta<1$ and $T>0$ such that for any $t>T$,
\begin{equation*}
\lVert \hat{f}\big(P^{\Theta_l}(t)\big)-\hat{P}^* \rVert_{\infty} \le \eta \lVert P^{\Theta_l}(t)-\hat{P}^* \rVert_{\infty}.
\end{equation*}
Therefore, $P^{\Theta_l}(t)\rightarrow \hat{P}^*$ as $t\rightarrow \infty$. 

$\text{ }$

(4) Case 4:
\begin{equation*}
\begin{split}
P^{\Theta_l}(t+1) & = \diag(\bm{\alpha})\tilde{A}P^{\Theta_l}(t)\\
& \quad + \big( I-\diag(\bm{\alpha}) \big) P^{\Theta_l}(t)\Delta_l\\
& \quad + \big( I-\diag(\bm{\alpha}) \big)P^{\Lambda}(t)B_l.
\end{split}
\end{equation*}
for any $l\in \{ 1,2,\dots,m \}$. Therefore,
\begin{equation}\label{eq_social_self_combined_struc_consensus_sys}
P^{\Theta_l}(t+1)\vectorones[k_l]=MP^{\Theta_l}(t)\vectorones[k_l] + \bm{\phi}(t),
\end{equation}
where $M=\diag(\bm{\alpha})\tilde{A}+I-\diag(\bm{\alpha})$ is row-stochastic and primitive. The vector $\bm{\phi}(t)$ is a vanishing perturbation according to the proof for Case 2.

Let $\bm{x}(t)=P^{\Theta_l}(t)\vectorones[k_l]$ and $\bm{y}(t)=Q\bm{x}(t)$ with $Q$ defined in Lemma~\ref{lem_pairwise_diff_transform}. Let $\subscr{\bm{y}}{err}(t)=(y_1(t),y_2(t),\dots,y_{n-1}(t))^{\top}$, where $y_i(t)=x_{i+1}(t)-x_i(t)$ for any $i=1,2,\dots,n-1$. Then we have
\begin{equation*}
\bm{y}(t+1)=QMQ^{-1}\bm{y}(t)+Q\bm{\phi}(t).
\end{equation*}
Let $\bm{\varphi}(t)=\big( \varphi_1(t),\varphi_2(t),\dots,\varphi_{n-1}(t) \big)^{\top}$ with $\varphi_i(t)=\sum_j Q_{ij}\phi_j(t)$. $\bm{\varphi}(t)$ is also a vanishing perturbation and 
\begin{equation*}
\subscr{\bm{y}}{err}(t+1)=\subscr{M}{red}\subscr{\bm{y}}{err}(t) +\bm{\varphi}(t).
\end{equation*}
The equation above is an exponentially stable linear system with a vanishing perturbation. Since $\rho(\subscr{M}{red})<1$, $\subscr{\bm{y}}{err}\rightarrow \vectorzeros[n-1]$ as $t\rightarrow \infty$, which implies that $P^{\Theta_l}(t)\vectorones[k_l]\rightarrow \gamma \vectorones[n]$ and $\gamma_l$ depends on $M$, $B_l$, $P^{\Theta_l}(0)$ and $P^{\Lambda}(0)$. Moreover, $\sum_l \gamma_l=1$ since $P(t)\vectorones[R]=\vectorones[n]$. Following the same argument in the proof for Case 3, we obtain
\begin{equation*}
\lim_{t\rightarrow \infty}\bm{p}^{\Theta_l(i)}(t)=\gamma_l \bm{w}^{\top}(\Delta_l).
\end{equation*}
\endproof

\subsubsection{Interpretations of Theorem~\ref{thm_asym_behav_social_self_NCPM}}
Analysis on Case 1 to 4 leads to the following conclusions: 1) The probability of adopting any product in the transient subgraph eventually decays to zero; 2) For the  product-conversion graph with only on absorbing SCC $G(\Delta_1)$, the system's asymptotic product-adoption probability distribution only depends on $\bm{w}(\Delta_1)$. In this case, the self conversion dominates the competitive propagation process; 3) With multiple absorbing SCCs in the product-conversion graph, the initial condition $P(t)$ and the structure of the social network $G(\tilde{A})$ together determine the fraction each absorbing SCC eventually takes in the total probability 1; 4) In each absorbing SCC $G(\Delta_l)$, the asymptotic adoption probability for each product is proportional to its corresponding entry of $\bm{\Delta_l}$.

\subsection{Further simulation work}
\emph{a) Accuracy of the social-self NCPM solution:} Simulation results have been presented to compare the solution to the social-self NCPM with the solution to the original Markov chain model defined by Assumption~\ref{asmp_social_self_conv}. Let the matrix $\Delta$ take the following form
\begin{equation}\label{eq_simulation_Delta_combined}
\Delta=
\begin{bmatrix}
\Delta_1 & \vectorzeros & \vectorzeros\\
\vectorzeros & \Delta_2 & \vectorzeros\\
B_1 & B_2 & \Delta_0
\end{bmatrix}
=
\begin{bmatrix}
0.6 & 0.4 & 0 & 0\\
0.3 & 0.7 & 0 & 0\\
0 & 0 & 1 & 0\\
0 & 0.8 & 0 & 0.2
\end{bmatrix}.
\end{equation}

The Markov-chin solution is computed by the Monte Carlo method. In each sampling, $A$, $\bm{\alpha}$ and $P(0)$ are randomly generated and set identical for the Markov chain and the NCPM. The probability $p_{12}(t)$ is plotted for both models on different types of social networks, such as the complete graph, the Erd\H os-R\'enyi graph, the power-law graph and the star graph. As shown in Figure~\ref{fig_accuracy_NCPM} and Figure~\ref{fig_accuracy_NCPM_power_star}, the solution to the social-self NCPM nearly overlaps with the Markov-chain solution in every plot, due to the i.i.d self conversion process. 

\begin{figure}
\begin{center}
\subfigure[$n=5$, $p=1$]{\label{fig_difference_N=1_p=1} \includegraphics[width=.34\linewidth]{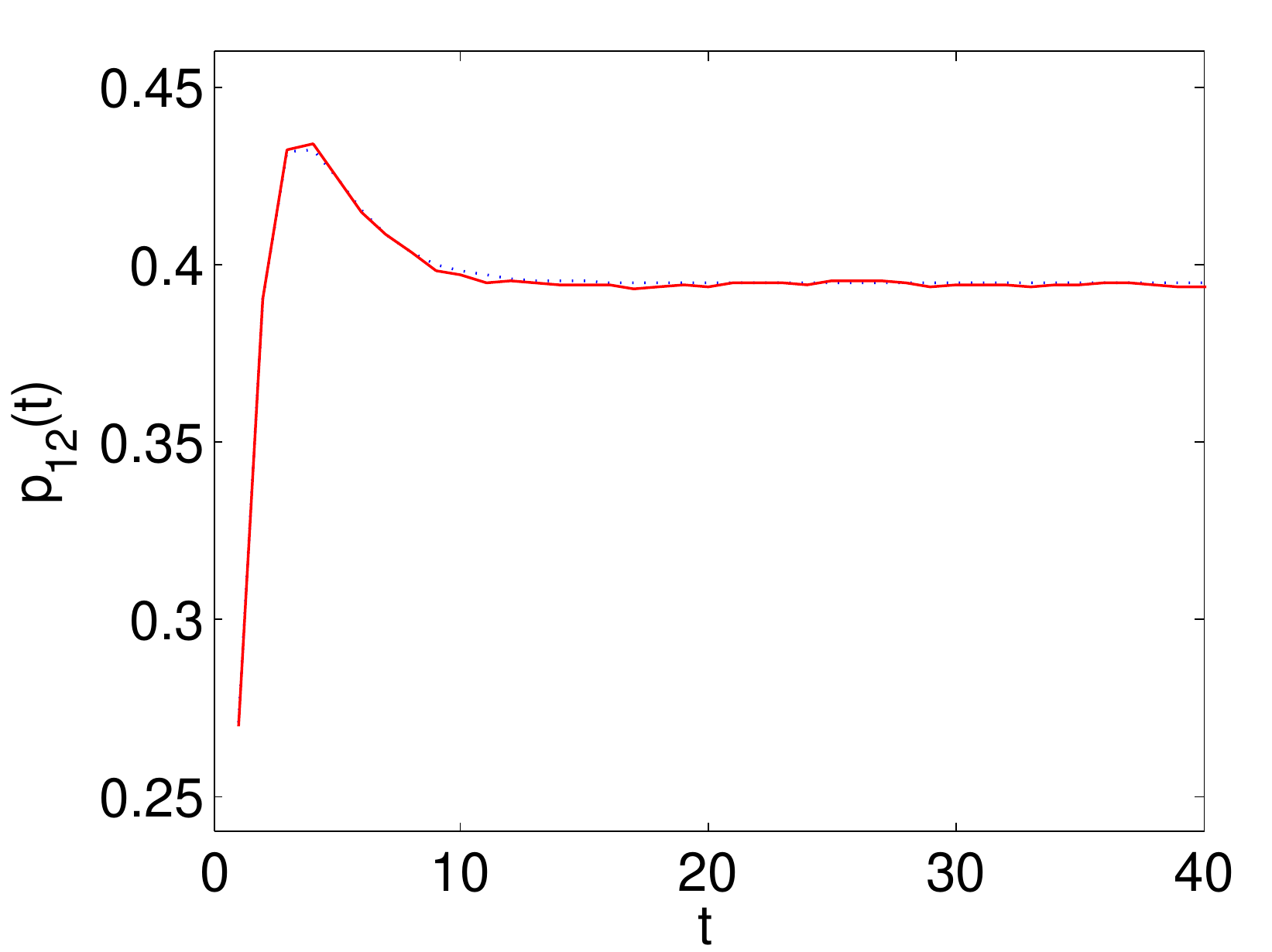}}
\subfigure[$n=10$, $p=0.5$]{\label{fig_difference_N=10_=05} \includegraphics[width=.35\linewidth]{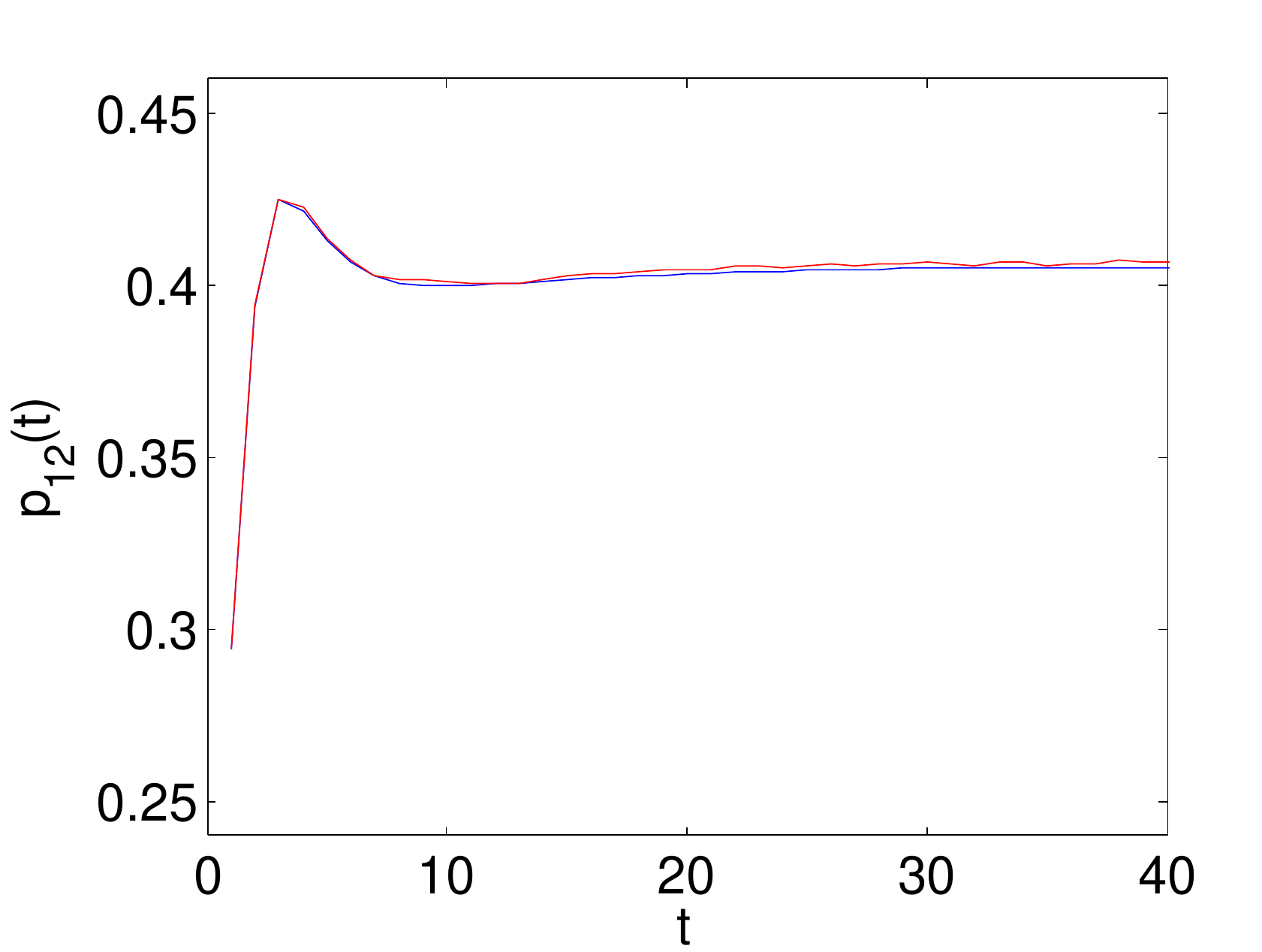}}
\subfigure[$n=50$, $p=0.1$]{\label{fig_difference_N=50_p=01} \includegraphics[width=.34\linewidth]{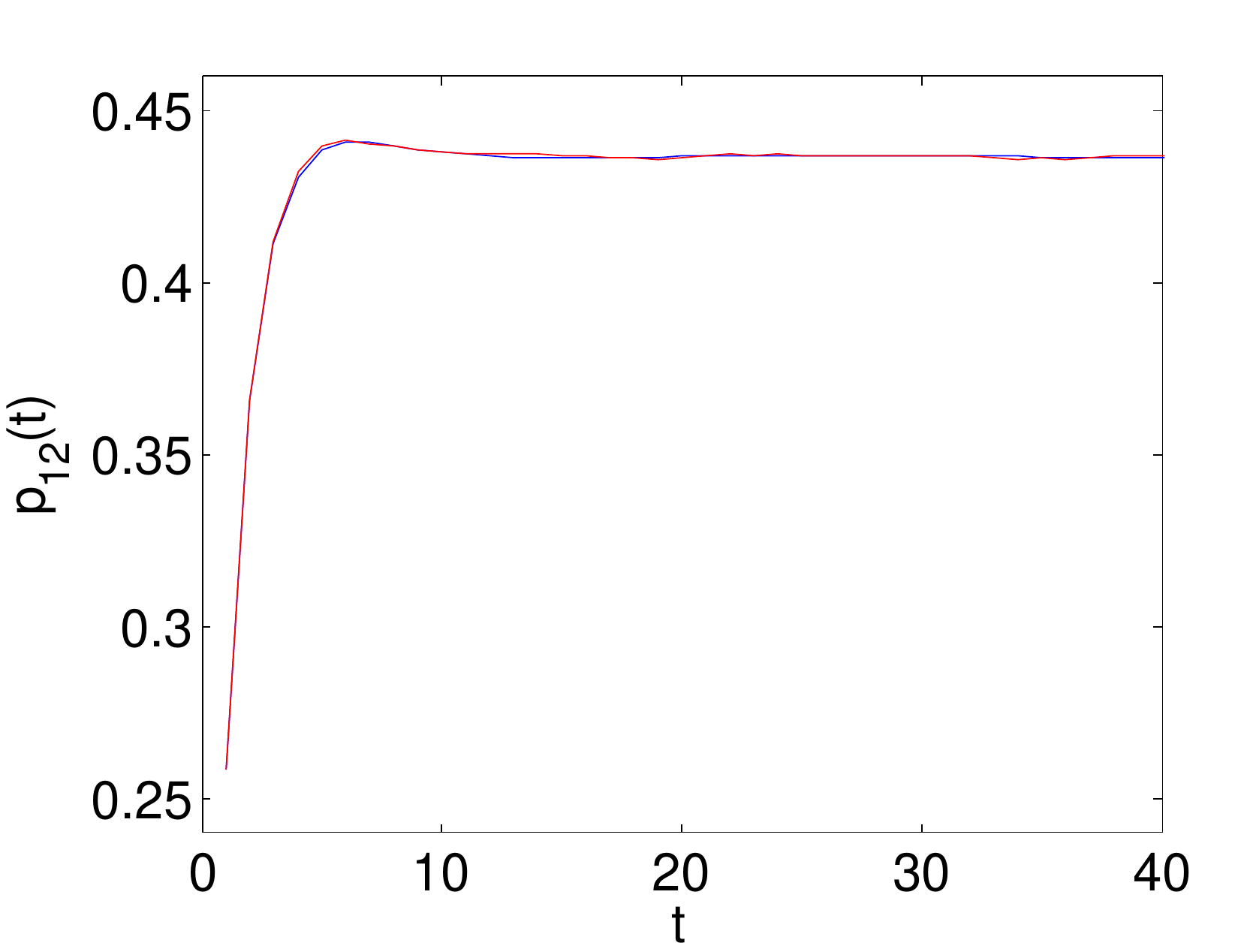}}
\subfigure[$n=50$, $p=1$]{\label{fig_difference_N=50_p=1} \includegraphics[width=.35\linewidth]{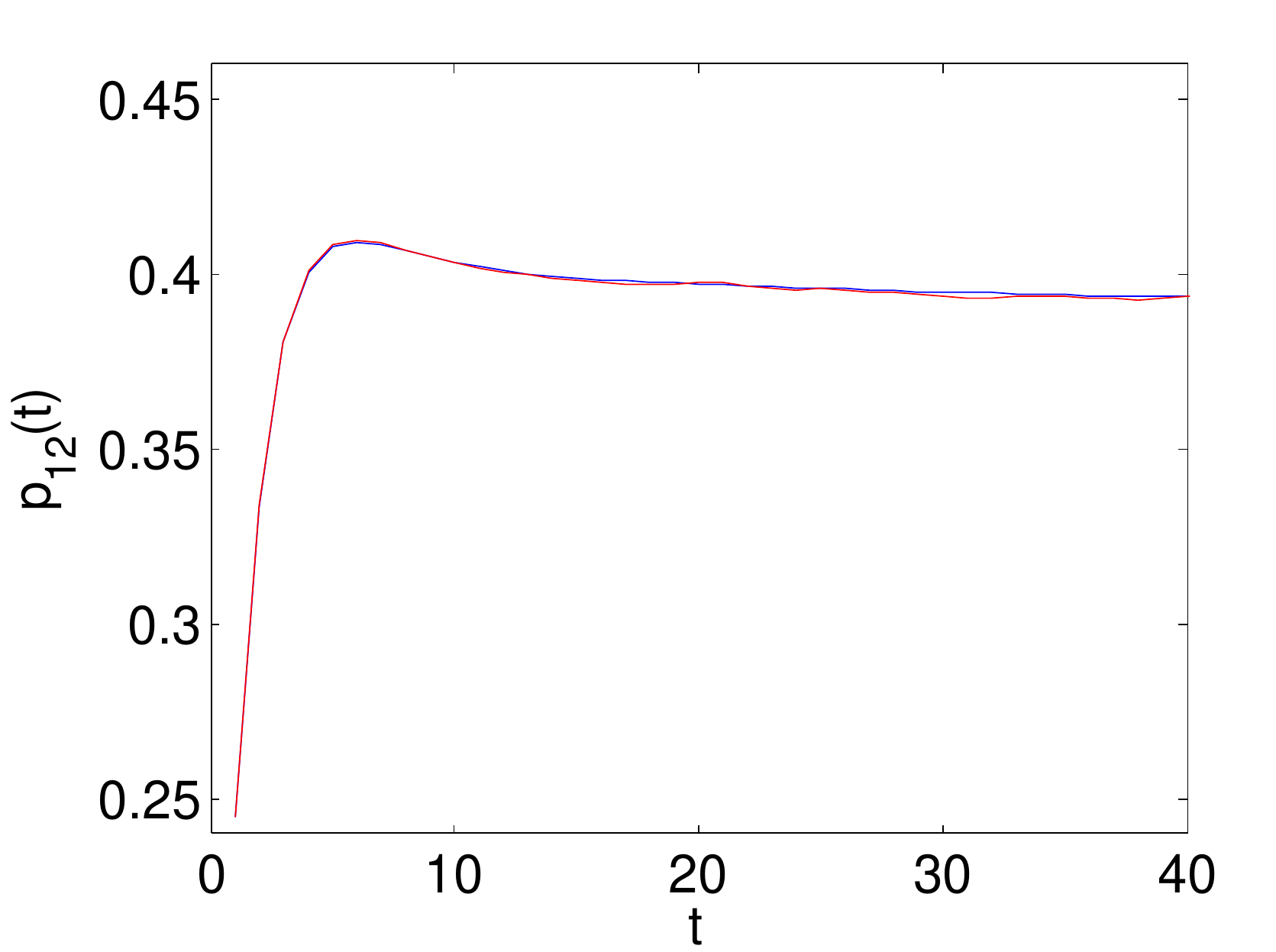}}
\caption{Difference between the solutions to the social-self NCPM (blue dash) and the original Markov-chain model (red) in complete graphs or Erd\H os-R\'enyi graphs.}
\label{fig_accuracy_NCPM}
\end{center}
\end{figure}

\begin{figure}
\begin{center}
\subfigure[power-law graph]{\label{fig_difference_power_law} \includegraphics[width=.39\linewidth]{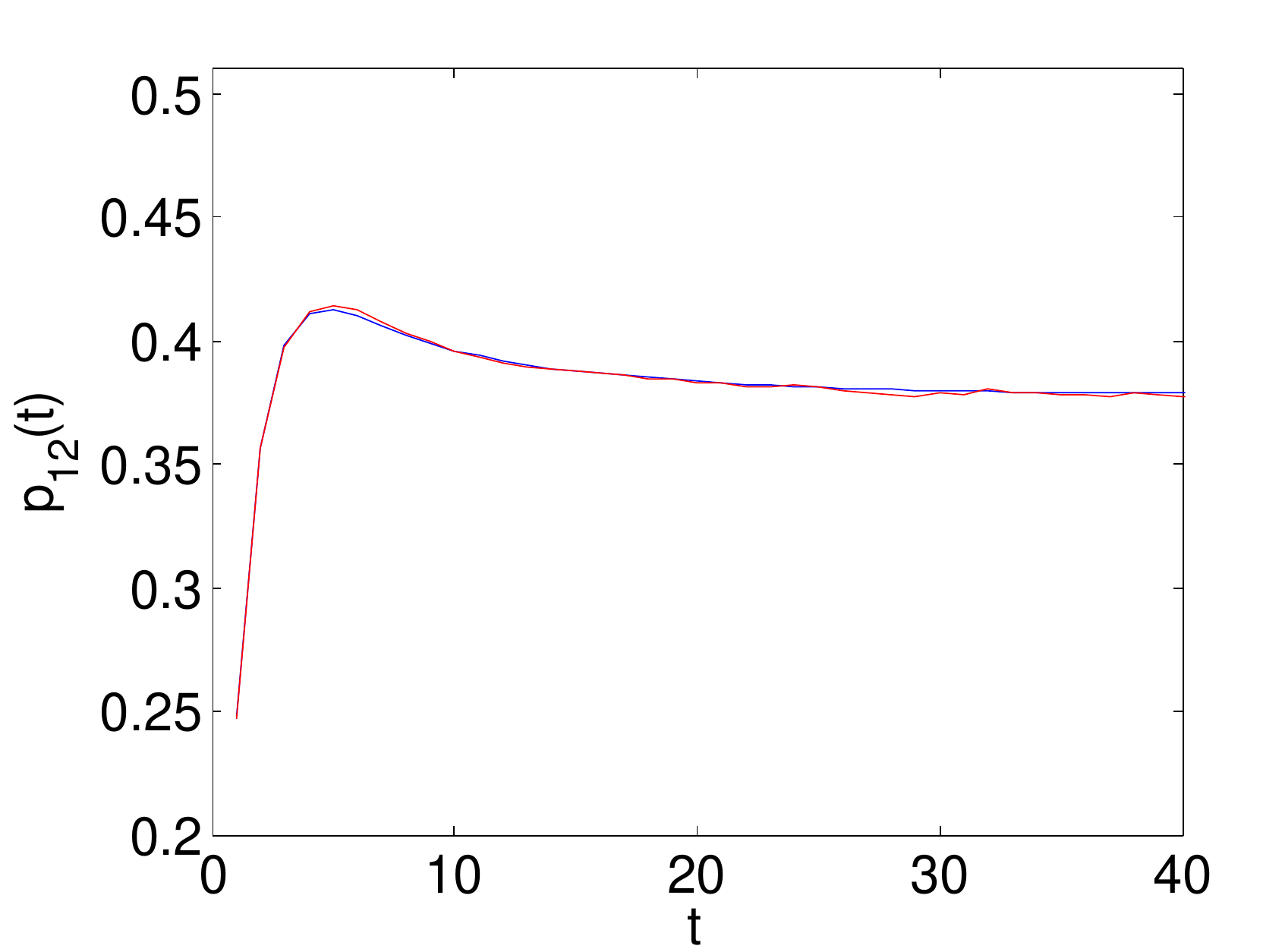}}
\subfigure[star graph]{\label{fig_difference_star_graph} \includegraphics[width=.40\linewidth]{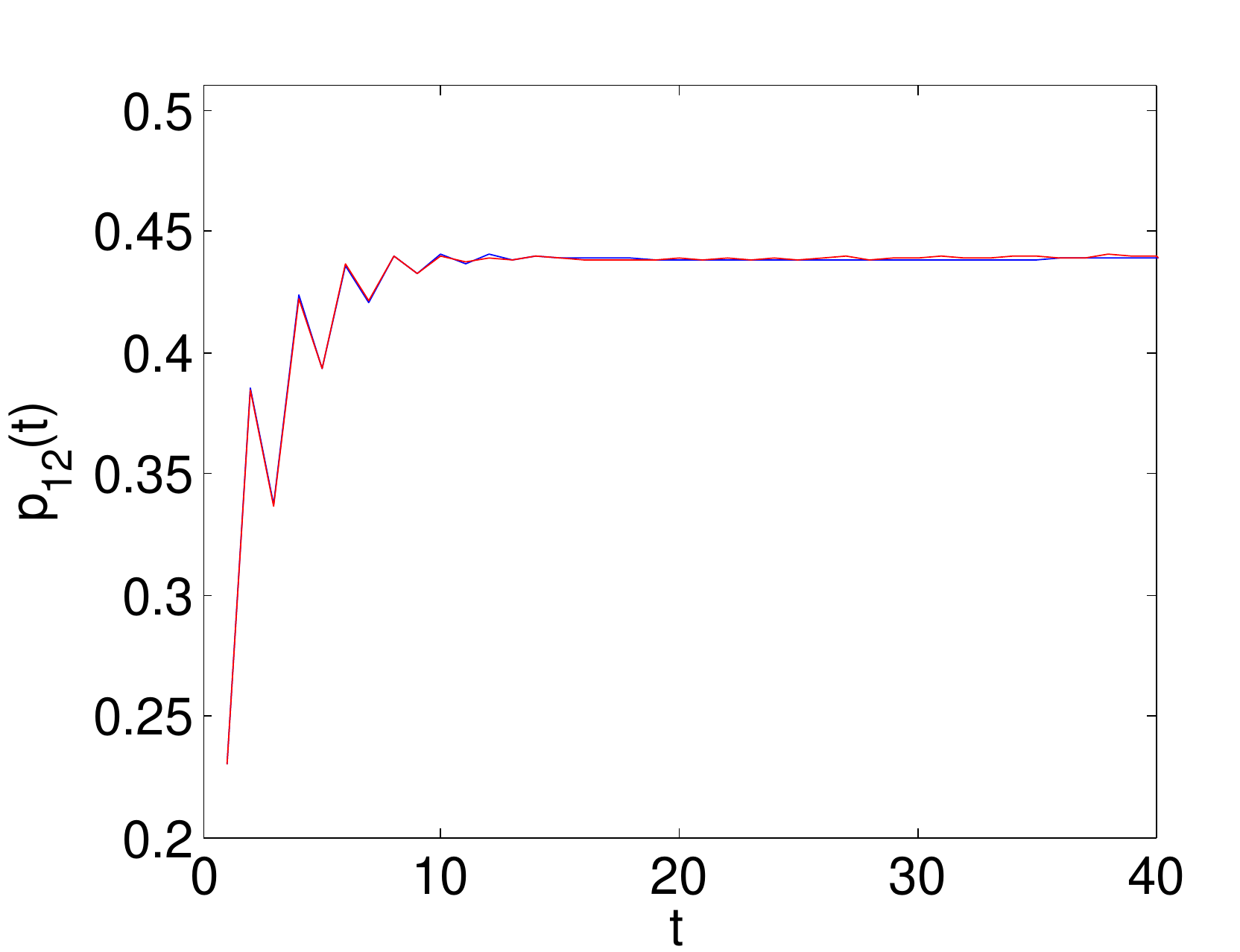}}
\caption{Difference between the solutions to the social-self NCPM (blue dash) and the original Markov-chain model (red) in the power-law graph and the star graph. The power-law graph has $100$ nodes, with the degree distribution $p(k)=1010k^{-2.87}$. The star graph consists of 10 nodes with node 1 as the center.}
\label{fig_accuracy_NCPM_power_star}
\end{center}
\end{figure}

\emph{b) Asymptotic behavior of the Markov chain model}
In Figure~\ref{fig_asym_Markov_two_structures} and Figure~\ref{fig_asym_Markov_combined}, all the trajectories $p_{ir}(t)$, for the Markov-chain model on an Erd\H os-R\'enyi graph with $n=5$, $p=0.4$ and randomly generated $\bm{\alpha}$, are computed by the Monte Carlo method. Figure~\ref{fig_asym_behav_Markov_transient_absorbing} corresponds to the structure of the product-conversion graph defined by Case 4 in Definition~\ref{def_foure_cases_product_conv_graph} with
\begin{equation*}
\Delta_1=
\begin{bmatrix}
0.6 & 0.4\\
0.3 & 0.7
\end{bmatrix},
\Delta_2=1,
\Delta_0=0.2,
B=[0\text{ }0.8\text{ }0].
\end{equation*} 
The transient subgraph is only connected to SCC $\Theta_1$ and the intial adoption probability for $H_3$ is 0. Figure~\ref{fig_asym_behav_Markov_isolated} corresponds to the structure of the product-conversion graph defined by Case 3 in Defintion~\ref{def_foure_cases_product_conv_graph} with
\begin{equation*}
\Delta=
\begin{bmatrix}
\Delta_1 & \vectorzeros\\
\vectorzeros & \Delta_2
\end{bmatrix},
\Delta_1=
\begin{bmatrix}
0.6 & 0.4\\
0.3 & 0.7
\end{bmatrix},
\Delta_2=
\begin{bmatrix}
0.5 & 0.5 \\
0.1 & 0.9
\end{bmatrix}.
\end{equation*}
The simulation results shows that, in these two cases the Markov-chain solutions converge exactly to the values indicated by the social-self NCPM, regardless of the initial condition. The matrix $\Delta$ used in Figure~\ref{fig_asym_Markov_combined} is given by equation~\eqref{eq_simulation_Delta_combined}. As illustrated by Figure~\ref{fig_asym_Markov_combined}, the asymptotic adoption probabilities vary with the initial condition in the Markov-chain model, in consistence with the results of Theorem~\ref{thm_asym_behav_social_self_NCPM}.

\begin{figure}
\begin{center}
\subfigure[multi-SCC + transient subgraph]{\label{fig_asym_behav_Markov_transient_absorbing} \includegraphics[width=.40\linewidth]{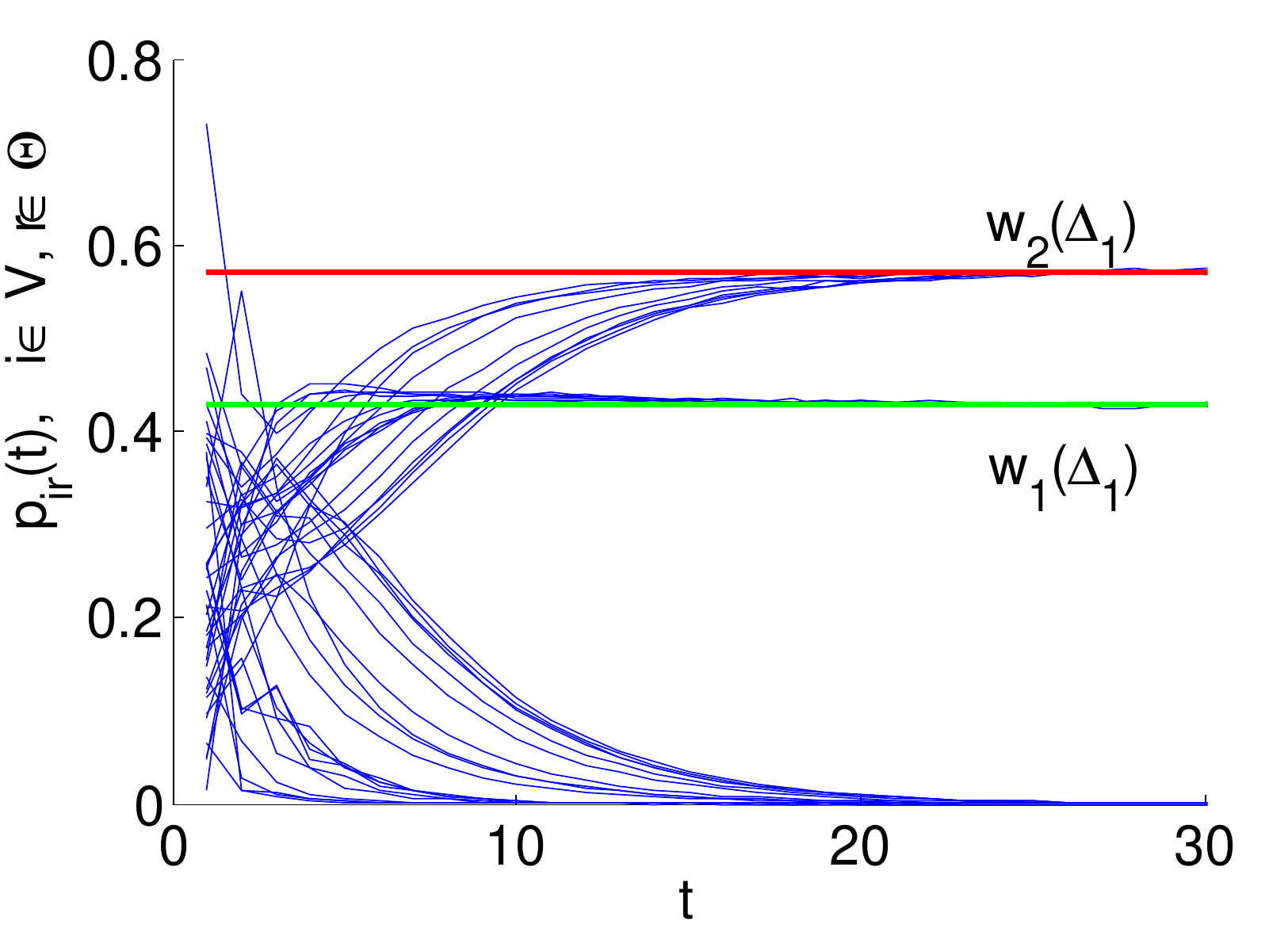}}
\subfigure[multi-SCC]{\label{fig_asym_behav_Markov_isolated} \includegraphics[width=.40\linewidth]{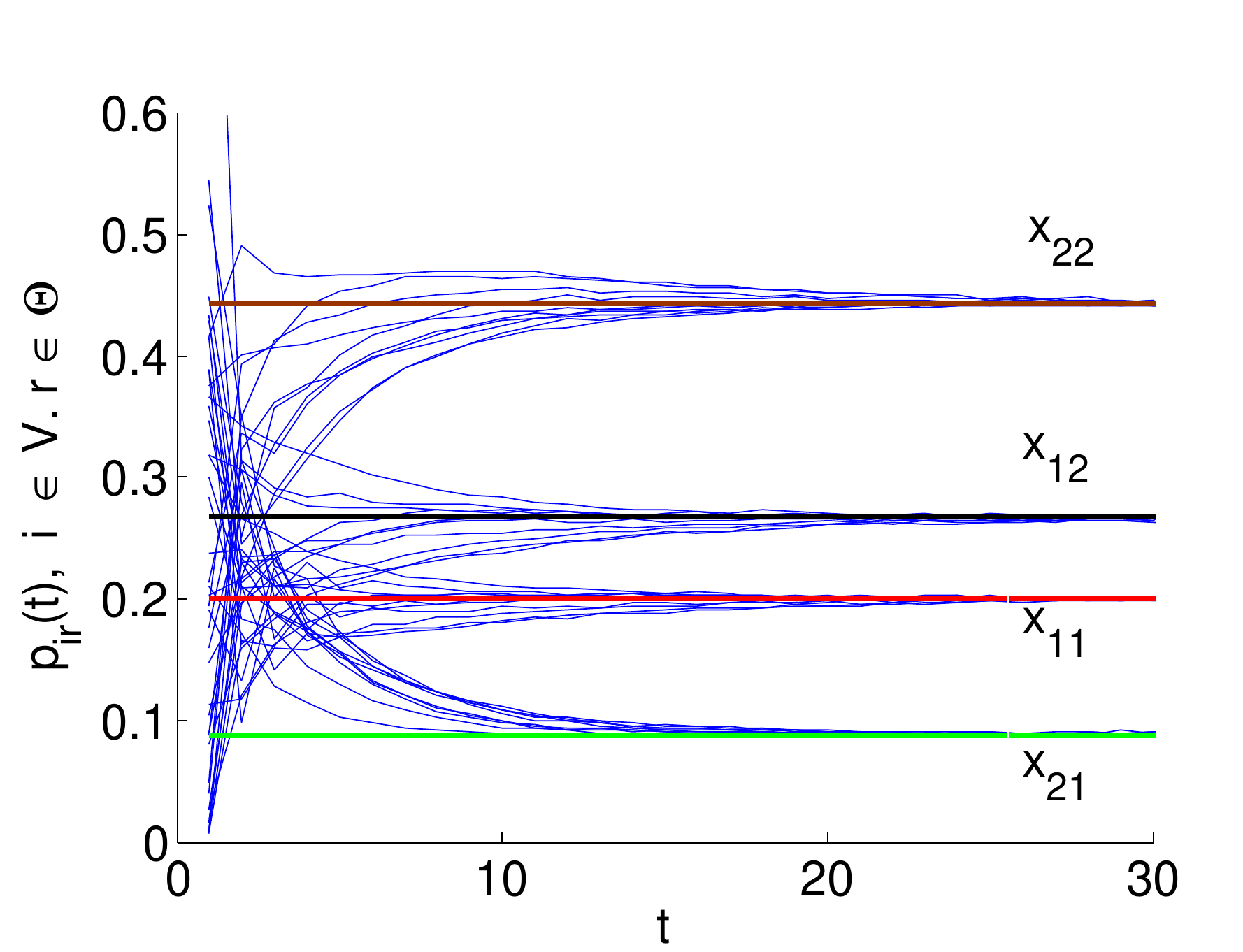}}
\caption{Asymptotic behavior of the Markov chain model with the production-conversion graphs defined by Case 3 or Case 4 in Defintion~\ref{def_foure_cases_product_conv_graph}. Every curve in this plot is a trajectory $p_{ir}(t)$ for $i \in V$ and $r\in \Theta$. Here $x_{lr}=\bm{w}^{\top}(M)P^{\Theta_l}(0)\vectorones[k_l]w_r(\Delta_l)$.}
\label{fig_asym_Markov_two_structures}
\end{center}
\end{figure}

\begin{figure}
\begin{center}
\subfigure[initial condition 1]{\label{fig_asym_behav_Markov_combined_1} \includegraphics[width=.40\linewidth]{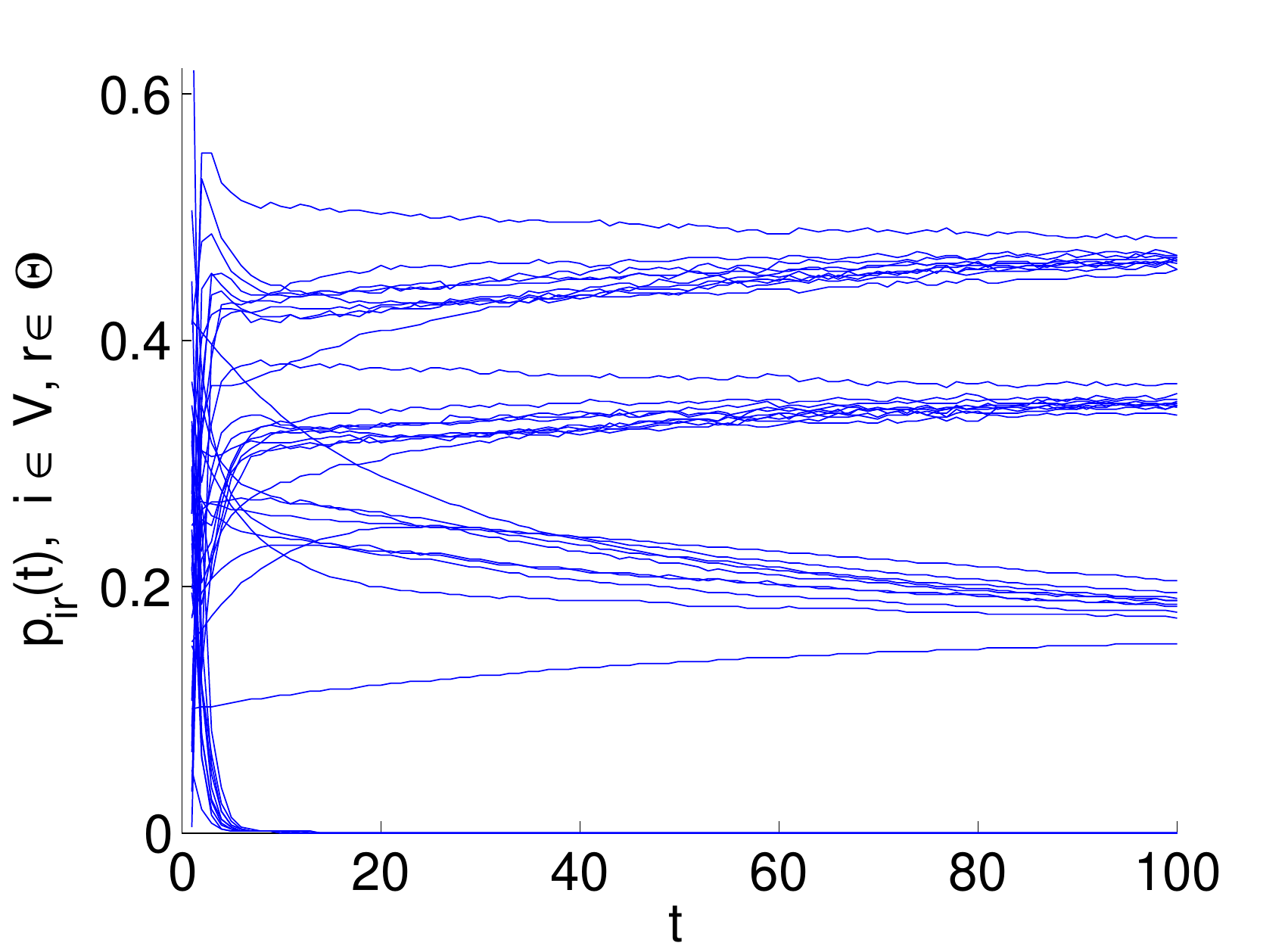}}
\subfigure[initial condition 2]{\label{fig_asym_behav_Markov_combined_2} \includegraphics[width=.40\linewidth]{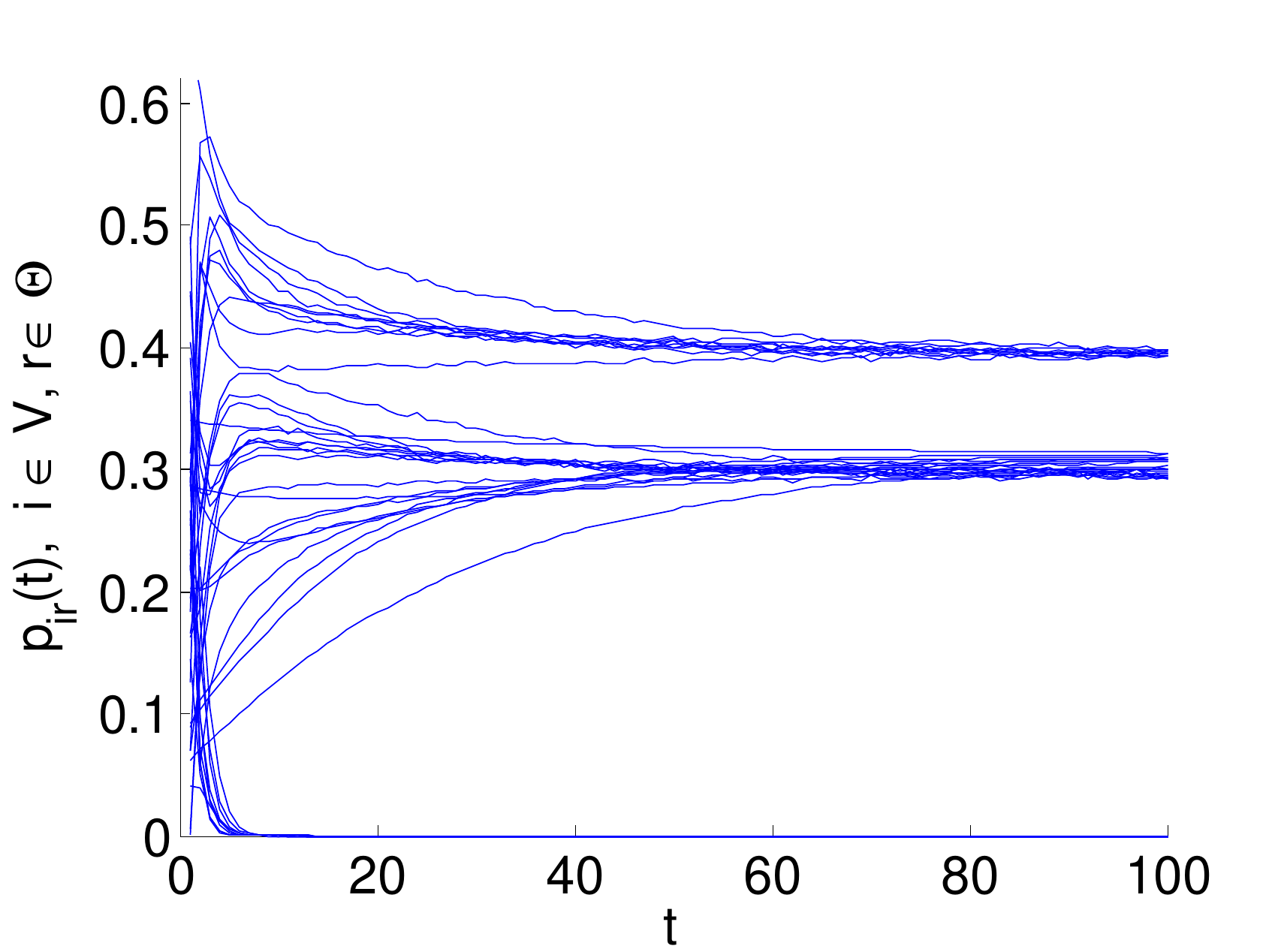}}
\caption{Asymptotic behavior of the Markov chain model with the production-conversion graph consisting of multiple SCCs and a transient subgraph. Every curve in this plot is a trajectory $p_{ir}(t)$ for $i \in V$ and $r\in \Theta$.}
\label{fig_asym_Markov_combined}
\end{center}
\end{figure}

\section{Analysis on the Self-social Network Competitive Propagation Model}
In this section we discuss the network competitive propagation model based on Assumption~\ref{asmp_self_social_conv}, i.e, the case in which self conversion occurs before social conversion at each time step. Similar to what we have done in the last section, firstly we propose an approximation model, referred to as the \emph{self-social network competitive propagation model} (self-social NCPM), and then analyze the dynamical properties of this approximation model.

The following theorem based on Approximation~\ref{indep_approx} gives the matrix form of the self-social NCPM.
\begin{theorem}[Self-social NCPM]\label{thm_self_social_mf}
Consider the competitive propagation model based on Assumption~\ref{asmp_self_social_conv}, with the social network and the product-conversion graph represented by their adjacency matrices $\tilde{A}$ and $\Delta$ respectively. The probability $p_{ir}(t)$ satisfies
\begin{equation*}
\begin{split}
p_{ir}(t+1)-p_{ir}(t) & = \sum_{s\neq r}\big( \delta_{sr}p_{is}(t)-\delta_{rs}p_{ir}(t) \big)\\
& \quad + \sum_{s\neq r}\delta_{ss}\alpha_i \sum_{j=1}^n \tilde{a}_{ij}p_{is}(t)P_{ji}^{rs}(t)\\
& \quad - \sum_{s\neq r} \delta_{rr}\alpha_i \sum_{j=1}^n \tilde{a}_{ij}p_{ir}(t)P_{ji}^{sr}(t),
\end{split}
\end{equation*}
for any $i\in V$ and $r\in \Theta$. Applying the independence assumption, the matrix form of the self-social NCPM is 
\begin{equation}\label{eq_mf_self_social_matrix}
\begin{split}
P(&t+1)\\
& = P(t)\Delta + \diag(\bm{\alpha})\diag\big( P(t)\bm{\delta} \big)\tilde{A}P(t)\\
& \quad -\diag(\bm{\alpha})P(t)\diag(\bm{\delta}),
\end{split}
\end{equation}
with $P(t)=(p_{ir}(t))_{n\times R}$ and $\bm{\delta} = (\delta_{11},\delta_{22},\dots,\delta_{RR})^{\top}$.
\end{theorem}
\smallskip
It is straightforward to check that, for any $P(t)\in S_{nR}(\vectorones[n])$, $P(t+1)$ is still in $S_{nR}(\vectorones[n])$. According to the Brower fixed point theorem, there exists at least one fixed point for the system~\eqref{eq_mf_self_social_matrix} in $S_{nR}(\vectorones[n])$. Since the nonlinearity of equation~\eqref{eq_mf_self_social_matrix} add much difficulty to the analysis of it, in the remaining part of this section we discuss the special case when $R=2$.

\subsection{Existence and uniqueness of the fixed point}
For simplicity, in this section, let $\bm{p}(t)=\bm{p}_2(t)=\big( p_{12}(t),p_{22}(t),\dots,p_{n2}(t) \big)^{\top}$. Without loss of generality, we always assume that $\delta_{22}\ge \delta_{11}$. Define the map $h:\mathbb{R}^n\rightarrow \mathbb{R}^n$ by
\begin{equation}\label{eq_self_social_two_prod_h_map}
\begin{split}
h(\bm{x}) & = \delta_{12}\vectorones[n] + (1-\delta_{12}-\delta_{21})\bm{x}\\
&\quad +\delta_{11}\diag(\bm{\alpha})\tilde{A}\bm{x}-\delta_{22}\diag(\bm{\alpha})\bm{x}\\
&\quad +(\delta_{22}-\delta_{11})\diag(\bm{\alpha})\diag(\bm{x})\tilde{A}\bm{x}.\\
\end{split}
\end{equation}
Then the self-social NCPM for $R=2$ is written as
\begin{equation}\label{eq_self_social_two_product_vector_form}
\bm{p}(t+1)=h(\bm{p}(t)),
\end{equation} 
and $\bm{p}_1(t)$ is computed by $\bm{p}_1(t)=\vectorones[n]-\bm{p}(t)$.

First we prove that the two-product self-social NCPM possesses a unique fixed point.
\begin{lemma}[Invariant domain of map $h$]\label{lem_self_social_h_invariant}
The map $h$ defined by equation~\eqref{eq_self_social_two_prod_h_map} is a continuous map from $[0,1]^n$ to $[0,1]^n$ itself. 
\end{lemma}
\smallskip
\begin{proof}
The map $h$ is polynomial and hence continuous. Firstly, we prove
that, for any $x\in [0,1]^n$, $h(\bm{x})\succeq \vectorzeros[n]$.
Since
\begin{equation*}
\begin{split}
h(\bm{x})&=\delta_{12}(\vectorones[n]-\bm{x})+\delta_{11}\diag(\bm{\alpha})\tilde{A}\bm{x}\\
&\quad +(1-\delta_{21})\bm{x}-\delta_{22}\diag(\bm{\alpha})\bm{x}\\
&\quad +(\delta_{22}-\delta_{11})\diag(\bm{\alpha})\diag(\bm{x})\tilde{A}\bm{x}\\
\end{split}
\end{equation*}
and 
\begin{equation*}
(1-\delta_{21})\bm{x}-\delta_{22}\diag(\bm{\alpha})\bm{x}\succeq (1-\delta_{21}-\delta_{22})\bm{x}=\vectorzeros[n],
\end{equation*}
the right-hand side of the expression of $h$ is
nonnegative. Therefore, for any $\bm{x}\in [0,1]^n$, $h(\bm{x})\succeq
\vectorzeros[n]$.

Secondly, we prove that for any $\bm{x}\in [0,1]^n$, $h(\bm{x})\preceq
\vectorones[n]$. Recall that $x_{-i}=(\tilde{A}\bm{x})_i=\sum_j
\tilde{a}_{ij}x_j$. That is, $x_{-i}$ is the weighted average of all
the $x_j$'s except $x_i$ and the value of $x_{-i}$ does not depend on
$x_i$ since $\tilde{a}_{ii}=0$. Moreover, since $\sum_j
\tilde{a}_{ij}=1$ for any $i\in V$, $x_{-i}$ is also in the interval
$[0,1]$.  According to equation~\eqref{eq_self_social_two_prod_h_map},
rewrite the $i$-th entry of $h(\bm{x})$ as
\begin{equation*}
h(\bm{x})_i=\delta_{12}+\delta_{11}\alpha_i x_{-i}+\eta_i x_i,
\end{equation*}
where
$\eta_i=1-\delta_{12}-\delta_{21}-\delta_{22}\alpha_i+(\delta_{22}-\delta_{11})\alpha_i
x_{-i}$. The maximum value of $\eta_i$ is
$1-\delta_{12}-\delta_{21}-\delta_{11}\alpha_i$, obtained when
$x_{-i}=1$. Therefore,
\begin{equation*}
\eta_i x_i \le \max (1-\delta_{12}-\delta_{21}-\delta_{11}\alpha_i,\text{ }0).
\end{equation*}
Then we have
\begin{equation*}
\begin{split}
h(\bm{x})_i & \le \delta_{12}+\delta_{11}\alpha_i+\max(1-\delta_{12}-\delta_{21}-\delta_{11}\alpha_i,\text{ }0)\\
& = \max(\delta_{22}, \text{ }\delta_{12}+\delta_{11}\alpha_i)<1. 
\end{split}
\end{equation*}
The inequality above leads to $h(\bm{x})\preceq \vectorones[n]$ for any $\bm{x}\in[0,1]^n$.
\end{proof}

\begin{theorem}[Existence and uniqueness of the fixed point for two-product self-social NCPM]\label{thm_self_social_fixed_point_two_product} The map $h$ defined by equation~\eqref{eq_self_social_two_prod_h_map}, with parameters $\delta_{11}$, $\delta_{12}$, $\delta_{21}$, $\delta_{22}$, $\alpha_1$, $\dots$, $\alpha_n$ all in the interval $(0,1)$ and $\delta_{22}\ge \delta_{11}$, possesses a unique fixed point point $\bm{p}^*$ in $[0,1]^n$.
\end{theorem} 
\smallskip
\begin{proof}
According to Lemma~\ref{lem_self_social_h_invariant}, for any $\bm{p}(t)\in [0,1]^n$, $\bm{p}(t+1)$ is still in $[0,1]^n$. According to the Brower fixed point theorem, there exists $\bm{p}^*$ such that $h(\bm{p}^*)=\bm{p}^*$. This concludes the proof of the existence of a fixed point.

Any fixed point of map $h$ should satisfy $h(\bm{p}^*)=\bm{p}^*$, that is,
\begin{equation}\label{eq_two_prod_fixed_point_equation}
\begin{split}
\vectorzeros[n] & = \delta_{12}\vectorones[n] + \delta_{11}\diag(\bm{\alpha})\tilde{A}\bm{p}^*\\
&\quad +(\delta_{22}-\delta_{11})\diag(\bm{\alpha})\diag(\bm{p}^*)\tilde{A}\bm{p}^*\\
&\quad -(\delta_{12}+\delta_{21})\bm{p}^*-\delta_{22}\diag(\bm{\alpha})\bm{p}^*.
\end{split}
\end{equation}
Therefore, 
\begin{equation*}
\begin{split}
\bm{p}^* & = \delta_{12}K^{-1}\vectorones[n]+\delta_{11}K^{-1}\diag(\bm{\alpha})\tilde{A}\bm{p}^*\\
& \quad + (\delta_{22}-\delta_{11})K^{-1}\diag(\bm{\alpha})\diag(\bm{p}^*)\tilde{A}\bm{p}^*,
\end{split}
\end{equation*}
where $K=(\delta_{12}+\delta_{21})I+\delta_{22}\diag(\bm{\alpha})$ is a positive diagonal matrix. Define a map $T:\mathbb{R}^n\rightarrow \mathbb{R}^n$ by 
\begin{equation}\label{eq_map_T}
\begin{split}
T(\bm{x}) & = \delta_{12}K^{-1} \vectorones[n] + \delta_{11}K^{-1}\diag(\bm{\alpha})\tilde{A}\bm{x}\\
& \quad +(\delta_{22}-\delta_{11})K^{-1}\diag(\bm{\alpha})\diag(\bm{x})\tilde{A}\bm{x}.
\end{split}
\end{equation}
The the existence and uniqueness of the fixed point in $[0,1]^n$ for the map $h$ is equivalent to the existence and uniqueness of the fixed point for the map $T$. Now we prove that $T$ has a unique fixed point in $[0,1]^n$ by showing that $T$ maps $[0,1]^n$ to $[0,1]^n$ and $T$ is a contraction map.

For any $\bm{x}$ and $\bm{y}\in [0,1]^n$, define the distance $d(\bm{x},\bm{y})$ by $d(\bm{x},\bm{y})=\lVert \bm{x}-\bm{y} \rVert_{\infty}$. Then $([0,1]^n,d)$ is a complete metric space.

According to equation~\eqref{eq_map_T}, since $K^{-1}$, $\diag(\bm{\alpha})$, $\tilde{A}$, $\delta_{22}-\delta_{11}$ and $\diag(\bm{x})$ are all nonnegative, for any $\bm{x}$, $\bm{y}\in [0,1]^n$ and $\bm{x}\preceq \bm{y}$, we have $T(\bm{x})\preceq T(\bm{y})$. Moreover,
\begin{equation*}
T(\vectorzeros[n])=\delta_{12}K^{-1}\vectorones[n]\succ \vectorzeros[n],
\end{equation*}  
and 
\begin{equation*}
\begin{split}
T(\vectorones[n]) & = \delta_{12}K^{-1}\vectorones[n]+\delta_{11}K^{-1}\bm{\alpha}+(\delta_{22}-\delta_{11})K^{-1}\bm{\alpha}\\
& = \delta_{12}K^{-1}\vectorones[n] +\delta_{22}K^{-1}\bm{\alpha}.
\end{split}
\end{equation*} 
Since 
\begin{equation*}
T(\vectorones[n])_i=\frac{\delta_{12}+\delta_{22}\alpha_i}{\delta_{12}+\delta_{21}+\delta_{22}\alpha_i}<1,
\end{equation*}
we have $T(\vectorones[n])\prec \vectorones[n]$. Therefore, for any $\bm{x}\in [0,1]^n$, $T(\bm{x})\in [0,1]^n$, i.e., $T$ maps $[0,1]^n$ to $[0,1]^n$.

Now we prove that $T$ is a contraction map. For any $\bm{x}$, $\bm{y}\in [0,1]^n$,
\begin{equation*}
\begin{split}
T(\bm{x})_i-T(\bm{y})_i & =\frac{\delta_{11}\alpha_i}{K_i} (x_{-i}-y_{-i}) \\
& \quad + \frac{(\delta_{22}-\delta_{11})\alpha_i}{K_i}(x_i x_{-i}-y_i y_{-i}).
\end{split}
\end{equation*}
Moreover,
\begin{equation*}
\begin{split}
\lvert x_{-i}-y_{-i} \rvert & \le \sum_{j=1}^n \tilde{a}_{ij} \lvert x_j-y_j \rvert \\
& \le (\sum_{j=1}^n \tilde{a}_{ij}) \max_{j} \lvert x_j-y_j \rvert  =\lVert \bm{x}-\bm{y} \rVert_{\infty},
\end{split}
\end{equation*}
and
\begin{equation*}
\begin{split}
\lvert x_i & x_{-i} - y_i y_{-i} \rvert\\
& \le \max \big( \max_i y_i^2 - \min_i x_i^2,\text{ }\max_i x_i^2-\min_i y_i^2 \big)\\
& \le 2\lVert \bm{x}-\bm{y} \rVert_{\infty}.
\end{split}
\end{equation*}
Therefore,
\begin{equation*}
\lvert T(\bm{x})_i - T(\bm{y})_i \rvert \le \epsilon_i \lVert \bm{x}-\bm{y} \rVert_{\infty},
\end{equation*}
where $\epsilon_i=\frac{(2\delta_{22}-\delta_{11})\alpha_i}{\delta_{12}+\delta_{21}+\delta_{22}\alpha_i}$. It is easy to check that $\epsilon_i<1$ for any $i\in V$ and $\epsilon_i$ does not depend on the $\bm{x}$ and $\bm{y}$. Let $\epsilon=\max_{i} \epsilon_i$. Then for any $\bm{x}$, $\bm{y} \in [0,1]^n$,
\begin{equation*}
\lVert T(\bm{x})-T(\bm{y}) \rVert_{\infty} \le \epsilon \lVert \bm{x}-\bm{y} \rVert_{\infty} \quad \text{with }\epsilon<1.
\end{equation*} 
Applying the Banach fixed point theorem, we know that the map $T$ possesses a unique fixed point $\bm{p}^*$ in $[0,1]^n$. In addition, for any $\bm{p}(0)$, the sequence $\{ \bm{p}(t) \}_{t\in \mathbb{N}}$ defined by $\bm{p}(t+1)=T\big(\bm{p}(t)\big)$ satisfies 
\begin{equation*}
\lim_{t\rightarrow \infty} \bm{p}(t)=\bm{p}^*.
\end{equation*}
\end{proof}

Theorem~\ref{thm_self_social_fixed_point_two_product} not only proves the existence and uniqueness of the fixed point $\bm{p}^*$ in $[0,1]^n$, but also implies some properties of $\bm{p}^*$ by introducing the map $T$. Two properties of $\bm{p}^*$ are given below.
\begin{corollary}[Lower and upper bound of the unique fixed point]\label{crly_self_social_bound_fixed_pt}
For the two-product self-social NCPM with $\delta_{22}\ge \delta_{11}$, the unique fixed point $\bm{p}^*$ in $[0,1]^n$ satisfies 
\begin{equation}
\frac{1}{2}\vectorones[n]\preceq \bm{p}^* \preceq \frac{\delta_{12}}{\delta_{12}+\delta_{21}}\vectorones[n]. 
\end{equation}
\end{corollary}

\begin{proof}
One can easily check that when $\delta_{22}\ge \delta_{11}$, $T(\frac{1}{2}\vectorones[n])\succeq \frac{1}{2}\vectorones[n]$ and $T(\frac{\delta_{12}}{\delta_{12}+\delta_{21}}\vectorones[n])\preceq \frac{\delta_{12}}{\delta_{12}+\delta_{21}}\vectorones[n]$. Therefore, $T$ maps $S=\{\bm{x}\in \mathbb{R}^n \,|\, \frac{1}{2}\vectorones[n] \preceq \bm{x} \preceq \frac{\delta_{12}}{\delta_{12}+\delta_{21}}\vectorones[n]\}$ to $S$ itself. Since $T$ is a contraction map, the unique fixed point $\bm{p}^*$ is in $S$. 
\end{proof}

This corollary has a meaningful interpretation. The condition $\delta_{22}\ge \delta_{11}$ is equivalent to $\delta_{12}\ge\delta_{21}$, which means that the nodes in state $H_1$ have a higher or equal tendency of converting to $H_2$ than the tendency of self conversion from $H_2$ to $H_1$. In this sense the product $H_2$ is advantageous to $H_1$ and therefore the fixed point is in favor of $H_2$, that is, $\bm{p}^*\ge \frac{1}{2}\vectorones[n]$.  

The following corollary gives an upper bound for the difference between $p^*_i$ and $p^*_{-i}=\sum_{j} \tilde{a}_{ij}p^*_j$.

\begin{corollary}[Difference between $p_i^*$ and $p_{-i}^*$]\label{crly_differenc_p_i_p_-i}
  For the two-product self-social NCPM with $\delta_{22}\ge
  \delta_{11}$, the unique fixed point satisfies, for any $i\in V$,
\begin{equation}\label{eq_diff_p_i_p_-i}
p_i^*-p_{-i}^* \le \frac{1-\frac{1}{2}\alpha_i}{\alpha_i} \frac{\delta_{22}-\delta_{11}}{\delta_{22}+\delta_{11}}.
\end{equation}
\end{corollary}

\begin{proof}
According to equation~\eqref{eq_two_prod_fixed_point_equation}, we have
\begin{equation*}
C_i p_i^* - C_{-i} p_{-i}^*=\delta_{12}-\delta_{12}p_i^*,
\end{equation*}
where $C_i=\delta_{21}+\delta_{22}\alpha_i$ and $C_{-i}=\delta_{11}\alpha_i +(\delta_{22}-\delta_{11})\alpha_i p_i^*$.

Firstly we point out that $C_i>C_{-i}$ because 
\begin{equation*}
C_i-C_{-i}=\delta_{21}+\alpha_i (\delta_{22}-\delta_{11})(1-p_i^*)>0.
\end{equation*}

Moreover, 
\begin{equation*}
\begin{split}
p_i^* - & p_{-i}^*\\ 
& =\frac{\delta_{12}-(\delta_{12}+C_i-C_{-i})p_i^*}{C_{-i}}\\
& =\frac{\delta_{12}-\big( \delta_{12}+\delta_{21}+\alpha_i (\delta_{22}-\delta_{11})(1-p_i^*) \big)p_i^*}{\delta_{11}\alpha_i+(\delta_{22}-\delta_{11})\alpha_i p_i^*}.
\end{split}
\end{equation*}
The right-hand side of the equation above with $\frac{1}{2}\le p_i^* \le \frac{\delta_{12}}{\delta_{12}+\delta_{21}}$ achieves its maximum value
\begin{equation*}
\frac{1-\frac{1}{2}\alpha_i}{\alpha_i} \frac{\delta_{22}-\delta_{11}}{\delta_{22}+\delta_{11}}
\end{equation*}
at $p_i^*=\frac{1}{2}$. This concludes the proof.
\end{proof} 

\subsection{Stability of the unique fixed point}
Notice that any sequence $\{\bm{p}(t)\}_{t\in \mathbb{N}}$ defined by $\bm{p}(t+1)=T(\bm{p}(t))$ converging to $\bm{p}^*$ does not necessarily lead to the global stability of $\bm{p}^*$ for the system~\eqref{eq_self_social_two_product_vector_form}. Further analysis is needed for the global or local stability of the same fixed point $\bm{p}^*$ in system~\eqref{eq_self_social_two_product_vector_form}.

First we consider a special case in which $\delta_{11}=\delta_{22}$.

\begin{proposition}[Global Stability for the two-product self-social NCPM with $\delta_{11}=\delta_{22}$]\label{prop_self_social_global_stability_delta_11=22}
For the two-product self-social NCPM given by equation~\eqref{eq_self_social_two_product_vector_form}, if $\delta_{11}=\delta_{22}$, then the system has a unique fixed point $\bm{p}^*=\frac{1}{2}\vectorones[n]$ and for any initial condition $\bm{p}(0)$, the solution sequence $\{\bm{p}(t)\}_{t\in \mathbb{N}}$ computed by $\bm{p}(t+1)=h\big( \bm{p}(t) \big)$ converges to $\bm{p}^*$ exponentially fast.
\end{proposition}
\smallskip
\begin{proof}
With $\delta_{11}=\delta_{22}$, the map $h$ becomes
\begin{equation*}
h\big( \bm{x} \big)=\bm{x}+\delta_{12}\vectorones[n]-2\delta_{12}\bm{x}+\delta_{11}\diag(\bm{\alpha})\big( \tilde{A}\bm{x}-\bm{x} \big).
\end{equation*}
One can easily check that $\bm{p}^*=\frac{1}{2}\vectorones[n]$ is a fixed point. According to Theorem~\ref{thm_self_social_fixed_point_two_product}, the fixed point is unique. Let $\bm{p}(t)=\bm{y}(t)+\frac{1}{2}\vectorones[n]$. Then the two-product self-social NCPM becomes
\begin{equation*}
\bm{y}(t+1)=M \bm{y}(t),
\end{equation*}
where $M=(1-2\delta_{12})I+\delta_{11}\diag(\bm{\alpha})\tilde{A}-\delta_{11}\diag(\bm{\alpha})$.

For any $i\in V$, if $1-2\delta_{12}-\delta_{11}\alpha_i\ge 0$, then the $i$-th absolute row sum of $M$ is equal to
\begin{equation*}
\sum_{j=1}^n \lvert M_{ij} \rvert = 1-2\delta_{12}-\delta_{11}\alpha_i+\delta_{11}\alpha_i=1-2\delta_{12}<1,
\end{equation*}
and, if $1-2\delta_{12}-\delta_{11}\alpha_i<0$, the $i$-th absolute row sum satisfies
\begin{equation*}
\sum_{j=1}^n \lvert M_{ij} \rvert =2\delta_{12}+\delta_{11}\alpha_i+\delta_{11}\alpha_i-1<1.
\end{equation*}
Since $\rho(M)\le \lVert M \rVert_{\infty}=\max_i \sum_{j=1}^n \lvert M_{ij} \rvert$, the spectral radius of $M$ is strictly less than $1$. The fixed point $\bm{p}^*=\frac{1}{2}\vectorones[n]$ is exponentially stable for any initial condition $\bm{p}(t)\in [0,1]^n$.
\end{proof}

For the case $\delta_{22}>\delta_{11}$, we give two propositions on the sufficient conditions, which are inequalities on the parameters $\bm{\alpha}$, $\delta_{11}$ and $\delta_{22}$, for the local stability and global stability respectively. By ``global stability'' we mean the stability of $\bm{p}^*$ for any $\bm{p}(0)\in [0,1]^n$.

\begin{proposition}[Sufficient condition on local stability for two-product self-social NCPM with $\delta_{22}>\delta_{11}$]\label{prop_local_stability_two_product}
Consider the two-product self-social NCPM~\eqref{eq_self_social_two_product_vector_form} on the  connected network represented by the adjacency matrix $A$. Suppose $\delta_{22}>\delta_{11}$. Then the unique fixed point $\bm{p}^*$ is locally stable as long as 
\begin{equation}\label{eq_self_social_two_prod_suff_cond_local_stability}
\alpha_i<\frac{8\delta_{11}\delta_{22}}{(\delta_{22}-\delta_{11})^2+8\delta_{11}\delta_{22}}.
\end{equation}
\end{proposition}  
\smallskip
\begin{proof}
Let $\bm{p}(t)=\bm{y}(t)+\bm{p}^*$. Then system~\eqref{eq_self_social_two_product_vector_form} becomes
\begin{equation*}
\bm{y}(t+1)=M\bm{y}+(\delta_{22}-\delta_{11})\diag(\bm{\alpha})\diag(\bm{y}(t))\tilde{A}\bm{y}(t).
\end{equation*}
The right-hand side of the equation above is a linear term $M\bm{y}(t)$ with a constant matrix $M$, plus a quadratic term. The matrix $M$ can be decomposed as $M=\tilde{M}-\delta_{12}I$ and $\tilde{M}=\tilde{M}^{(1)}+\tilde{M}^{(2)}$ is further decomposed into a diagonal matrix $\tilde{M}^{(1)}$ and a matrix $\tilde{M}^{(2)}$ in which all the diagonal entries are $0$.
Since 
\begin{equation*}
\begin{split}
\tilde{M}^{(1)} & = (1-\delta_{12})I-\delta_{22}\diag(\bm{\alpha})\\
& \quad + (\delta_{22}-\delta_{11})\diag(\bm{\alpha})\diag(\tilde{A}\bm{p}^*)
\end{split}
\end{equation*}
is a positive diagonal matrix, and 
\begin{equation*}
\tilde{M}^{(2)}=\delta_{11}\diag(\bm{\alpha})\tilde{A} + (\delta_{22}-\delta_{11})\diag(\bm{\alpha})\diag(\bm{p}^*)\tilde{A}
\end{equation*}
is a matrix with all the diagonal entries being zero and all the off-diagonal entries being nonnegative, the matrix $\tilde{M}=\tilde{M}^{(1)}+\tilde{M}^{(2)}$ is nonnegative and, thereby, the matrix $M$ is thus a Metzler matrix.


Since $\tilde{A}=\diag(\frac{1}{N_1},\frac{1}{N_2},\dots,\frac{1}{N_n})A$, the matrix $\tilde{M}$ can be written in the form $DA+E$, where $A$ is symmetric and $D$, $E$ are positive diagonal matrix. One can easily prove that all the eigenvalues of any matrix in the form $\tilde{M}=DA+E$ are real since $\tilde{M}$ is similar to the symmetric matrix $D^{\frac{1}{2}}(A+D^{-1}E)D^{\frac{1}{2}}$.

The local stability of $\bm{p}^*$ is equivalent to the inequality $\rho(M)<1$, which is in turn equivalent to the intersection of the following two conditions: $\lambda_{\max}(\tilde{M})<1+\delta_{12}$ and $\lambda_{\min}(\tilde{M})>-1+\delta_{12}$.


First we prove $\lambda_{\max}(\tilde{M})<1+\delta_{12}$. Since $A$ is irreducible and $\bm{\alpha}\succ \vectorzeros[n]$, $\bm{p}^*\succ \vectorzeros[n]$, we have $\tilde{M}_{ij}>0$ if and only if $a_{ij}>0$ for any $i\neq j$. In addition, $\tilde{M}_{ii}>0$ for any $i\in V$. Therefore, $\tilde{M}$ is irreducible, aperiodic and thus primitive. According to the Perron-Frobenius theorem, $\lambda_{\max}(\tilde{M})=\rho(\tilde{M})$. We have $\rho(\tilde{M})\le \lVert \tilde{M} \rVert_{\infty}$ and for any $i\in V$,
\begin{equation*}
\sum_{j} \lvert \tilde{M}_{ij} \rvert = 1-\delta_{21}+(\delta_{22}-\delta_{11})\big( \alpha_i(p_{-i}^*+p_i^*)-\alpha_i \big).
\end{equation*}
According to Corollary~\ref{crly_self_social_bound_fixed_pt}, for any $i\in V$,
\begin{equation*}
1-\delta_{21}\le \sum_{j} \lvert \tilde{M}_{ij} \rvert \le 1-\delta_{21} + \frac{(\delta_{12}-\delta_{21})^2}{\delta_{12}+\delta_{21}}\alpha_i < 1+\delta_{12}.
\end{equation*}
Therefore, 
\begin{equation*}
\lambda_{\max}(\tilde{M})\le 1-\delta_{21}+\frac{(\delta_{12}-\delta_{21})^2}{\delta_{12}+\delta_{21}}\alpha_i < 1+\delta_{12}.
\end{equation*}

Now we prove $\lambda_{\min}(\tilde{M})>-1+\delta_{12}$. According to the Gershgorin circle theorem,
\begin{equation*}
\lambda_{\min}(\tilde{M})\ge \min_i (\tilde{M}_{ii}-\sum_{j\neq i} \lvert \tilde{M}_{ij} \rvert).
\end{equation*}
For any $i\in V$,
\begin{equation*}
\begin{split}
\tilde{M}_{ii}- & \sum_{j\neq i}\lvert \tilde{M}_{ij} \rvert \\
& = 1-\delta_{21}-\delta_{22}\alpha_i+(\delta_{22}-\delta_{11})\alpha_i \sum_j \tilde{a}_{ij}p_j^*\\
& \quad -\delta_{11}\alpha_i - (\delta_{22}-\delta_{11})\alpha_i p_i^*\\
& = 1-\delta_{21}-\alpha_i (\delta_{22}+\delta_{11})\\
& \quad -\alpha_i (\delta_{22}-\delta_{11})(p_i^*-p_{-i}^*).
\end{split}
\end{equation*}
According to Corollary~\ref{crly_differenc_p_i_p_-i},
\begin{equation*}
p_i^*-p_{-i}^* \le \frac{1-\frac{1}{2}\alpha_i}{\alpha_i} \frac{\delta_{22}-\delta_{11}}{\delta_{22}+\delta_{11}}.
\end{equation*}
Moreover, inequality~\eqref{eq_self_social_two_prod_suff_cond_local_stability} is necessary and sufficient to
\begin{equation*}
\frac{1-\frac{1}{2}\alpha_i}{\alpha_i} \frac{\delta_{22}-\delta_{11}}{\delta_{22}+\delta_{11}} < \frac{1-\alpha_i}{\alpha_i} \frac{\delta_{22}+\delta_{11}}{\delta_{22}-\delta_{11}}.
\end{equation*}
Therefore, 
\begin{equation*}
\begin{split}
\tilde{M}_{ii}- & \sum_{j\neq i}\lvert \tilde{M}_{ij} \rvert\\
& > 1-\delta_{21}-\alpha_i (\delta_{22}+\delta_{11})-(1-\alpha_i)(\delta_{22}+\delta_{11})\\
& = -1+\delta_{12},
\end{split}
\end{equation*}
for any $i\in V$. That is to say, the inequality~\eqref{eq_self_social_two_prod_suff_cond_local_stability} is sufficient to $\rho(M)<1$, i.e., the local stability of $\bm{p}^*$.
\end{proof}

From the proof we know that, around the unique fixed point, the linearized system is $\bm{y}(t+1)=M\bm{y}(t)$, where $M$ is a Metzler matrix and is Hurwitz stable. Usually the Metzler matrices are presented in continuous-time network dynamics models, for example, the epidemic spreading model~\cite{AF-AI-GS-JJT:07,AK-TB-BG:14}. In the proof of Proposition~\ref{prop_local_stability_two_product}, we provide an example for which the Metzler matrix appears in a stable discrete-time system. 

In the proposition below, a sufficient condition on the global stability of $\bm{p}^*$ is given.
\begin{proposition}[Sufficient condition on global stability for two-product self-social NCPM with $\delta_{22}>\delta_{11}$]\label{prop_global_stability_sufficient}
Consider the two-product self-social NCPM on the connected network with adjacency matrix $A$. Suppose $\delta_{22}>\delta_{11}$. If
\begin{equation}\label{eq_self_social_two_prod_suff_cond_global_stability}
\alpha_i<\frac{\delta_{22}+\delta_{11}}{3\delta_{22}-\delta_{11}},
\end{equation}
then
\begin{enumerate}
\item for any initial condition $\bm{p}(0)\in [0,1]^n$, the sequence $\{\bm{p}(t)\}_{t\in \mathbb{N}}$ defined by $\bm{p}(t+1)=f\big( \bm{p}(t) \big)$ satisfies $\bm{p}(t)\rightarrow \bm{p}^*$ exponentially fast as $t\rightarrow \infty$;

\item and moreover, the convergence rate is  upper bounded by $\max_i \big( \max(\epsilon_i, K_i\epsilon_i+K_i-1) \big)$, where $\epsilon_i$ and $K_i$ are respectively defined by $\epsilon_i=(2\delta_{22}-\delta_{11})\alpha_i/K_i$, and $K_i=\delta_{12}+\delta_{21}+\delta_{22}\alpha_i$.
\end{enumerate}
\end{proposition} 
\smallskip
\begin{proof}
Observe that the maps $h$ and $T$ satisfy the following relation:
\begin{equation*}
h(\bm{x})=KT(\bm{x})+(I-K)\bm{x},
\end{equation*}
for any $\bm{x}\in [0,1]^n$, where $K=(\delta_{12}+\delta_{21})I+\delta_{22}\diag(\bm{\alpha})$. For any $\bm{x}$, $\bm{y}\in [0,1]^n$,
\begin{equation*}
\begin{split}
\lvert h(\bm{x})_i- & h( \bm{y})_i \rvert\\
& = \lvert K_i\big( T(\bm{x})_i-T(\bm{y})_i \big) + (1-K_i)(x_i-y_i) \rvert.
\end{split}
\end{equation*}
We estimate the upper bound of $\lvert h(\bm{x})_i-h(\bm{y})_i \rvert$ in terms of $\lVert \bm{x}-\bm{y} \rVert_{\infty}$ in two cases.

\emph{Case 1:} $\delta_{12}+\delta_{21}+\delta_{22}\alpha_i<1$ for any $i$, i.e, $\alpha_i<\frac{\delta_{11}}{\delta_{22}}+1-\frac{1}{\delta_{22}}$ for any $i$. First we point out that
\begin{equation*}
\frac{\delta_{11}}{\delta_{22}}+1-\frac{1}{\delta_{22}}<\frac{\delta_{11}+\delta_{22}}{3\delta_{22}-\delta_{11}}
\end{equation*}
always holds as long as $\delta_{11}<\delta_{22}$. Then recall that, for any $\bm{x}$, $\bm{y}\in [0,1]^n$,
\begin{equation*}
\lvert T(\bm{x})_i-T(\bm{y})_i \rvert \le \epsilon_i \lVert \bm{x}-\bm{y} \rVert_{\infty},
\end{equation*}
where $\epsilon_i=\frac{(2\delta_{22}-\delta_{11})\alpha_i}{K_i}<1$. Therefore,
\begin{equation*}
\lvert h(\bm{x})_i-h(\bm{y})_i \rvert \le (K_i\epsilon_i+1-K_i) \lVert \bm{x}-\bm{y} \rVert_{\infty},
\end{equation*}
for any $i \in V$. The coefficient $K_i\epsilon_i+1-K_i$ is always strictly less than $1$ because it is a convex combination of $\epsilon_i<1$ and $1$. Therefore, $h$ is a contraction map.

\emph{Case 2:} There exists some $i$ such that $\delta_{12}+\delta_{21}+\delta_{22}\alpha_i \ge 1$, i.e., $\alpha_i \ge \frac{\delta_{11}}{\delta_{22}}+1-\frac{1}{\delta_{22}}$. In this case, for any such $i$,
\begin{equation*}
\lvert h(\bm{x})_i-h(\bm{y})_i \rvert \le (K_i\epsilon_i+K_i-1) \lVert \bm{x}-\bm{y} \rVert_{\infty}.
\end{equation*}
If 
\begin{equation*}
\alpha_i<\frac{\delta_{11}+\delta_{22}}{3\delta_{22}-\delta_{11}},
\end{equation*}
then we have
\begin{equation*}
\begin{split}
K_i\epsilon_i+K_i-1 & = (3\delta_{22}-\delta_{11})\alpha_i + \delta_{12}+\delta_{21}-1\\
& < \delta_{11}+\delta_{22}+\delta_{12}+\delta_{21}-1\\
& = 1.
\end{split}
\end{equation*}
Therefore, $h$ is also a contraction map.

Combining Case 1 and Case 2 we conclude that if $\alpha_i<\frac{\delta_{11}+\delta_{22}}{3\delta_{22}-\delta_{11}}$ for any $i\in V$, then $h$ is a contraction map. According to Lemma~\ref{lem_self_social_h_invariant}, $h$ maps $[0,1]^n$ to $[0,1]^n$. Therefore, according to the Banach fixed point theorem, for any initial condition $\bm{p}(0) \in [0,1]^n$, the solution $\bm{p}(t)$ converges to $\bm{p}^*$ exponentially fast and the convergence rate is upper bounded by $\max_i \big( \max(\epsilon_i,K_i\epsilon_i+K_i-1) \big)$.
\end{proof}

Figure~\ref{fig_self_soc_loc_glo_cond_stability} plots the right-hand sides of inequalities~\eqref{eq_self_social_two_prod_suff_cond_local_stability} and~\eqref{eq_self_social_two_prod_suff_cond_global_stability}, respectively, as functions of the ratio $\frac{\delta_{11}}{\delta_{22}}$, for the case when $0<\frac{\delta_{11}}{\delta_{22}}<1$. One can observe that, for a large range of $\frac{\delta_{qq}}{\delta_{22}}$, the sufficient condition we propose for the global stability is more conservative than the sufficient condition for the local stability.

\begin{figure}
\begin{center}
\includegraphics[width=.6\linewidth]{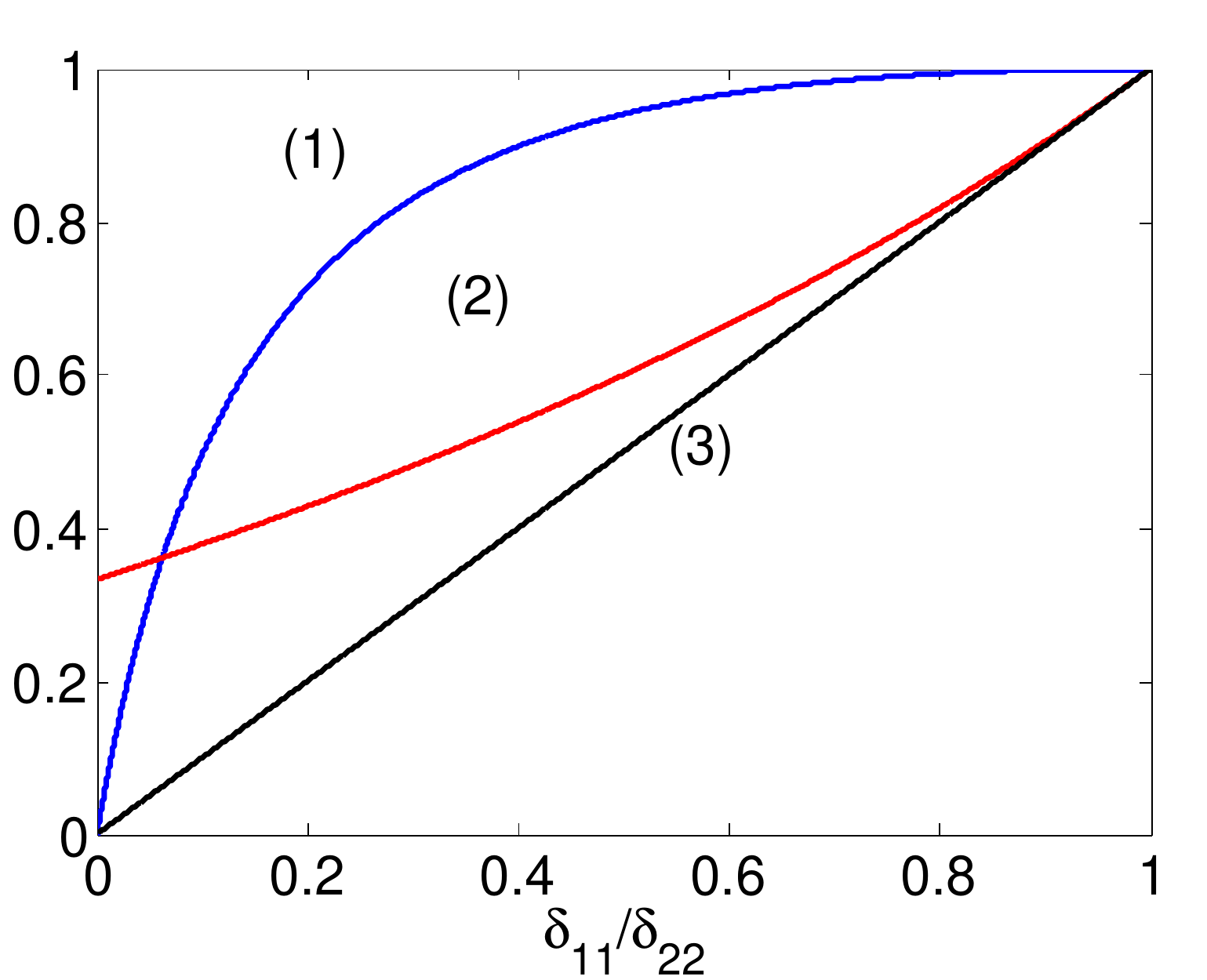}
\caption{This figure illustrates how the conditions for the local stability and global stability change with the ratio $\delta_{11}/\delta_{22}$. Curve (1) is $8\delta_{11}\delta_{22}/((\delta_{22})-\delta_{11})^2+8\delta_{11}\delta_{22})$, i.e, corresponding to the condition for local stability. Curve (2) is $(\delta_{22}+\delta_{11})/(3\delta_{22}-\delta_{11})$, corresponding to the condition for global stability. Curve (3) is $\delta_{11}/\delta_{22}$.}
\label{fig_self_soc_loc_glo_cond_stability}
\end{center}
\end{figure}

One major difference between the self-social and the social-self NCPM in the asymptotic property is that, in the self-social NCPM, every individual's state probability distribution is not necessarily identical. Moreover, distinct from the social-self NCPM, for any of the four cases of $G(\Delta)$ defined in Definition~\ref{def_foure_cases_product_conv_graph}, the asymptotic behavior of the self-social NCPM depends on not only the structure of $G(\Delta)$, but also the structure of the social network $G(\tilde{A})$ and the individual open-mindedness $\bm{\alpha}$. 


\section{Non-cooperative Multi-stage Competitive Propagation Games}

In this section, based on the social-self NCPM given by equation~\eqref{eq_mf_social_self_matrix}, we propose two types of non-cooperative, multi-player and multi-stage games. The players are the $R$ companies, each of which has a product competing in the social network. At each time step, based on the system's current product-adoption probability distribution, all the companies make decisions on the allocation of their investments, with limited budget, to maximize the probability that their products are adopted after the current time step. In the first subsection we discuss the model in which each company can invest both on seeding, e.g., advertisement and promotion, and on the product's quality; Then, in the second subsection, we discuss the model in which the products' quality is fixed and the companies can only invest on seeding.

All the notations in Table~\ref{table_notations} and the previous sections still apply and, in Table~\ref{table_notations_secV}, we introduce some additional notations and functions exclusively for this section.
\begin{table}[htbp]\caption{Notations and functions used in Section V}\label{table_notations_secV}
\begin{center}
\begin{tabular}{r p{5.8cm} }
\toprule
$X(t)$ & seeding matrix at time $t$. $X(t)=\big( x_{ir}(t) \big)_{n\times R}$, where $x_{ir}(t)\ge 0$ is company $r$'s investment on seeding for individual $i$. $\bm{x}_r(t)$ is the $r$-th column of $X(t)$ and $\bm{x}^{(i)}(t)$ is the $i$-th column of $X(t)$\\
$\bm{w}(t)$ & the quality investment vector at time $t$. $\bm{w}(t)\in \mathbb{R}^{R\times 1}$, and each entry $w_r(t)\ge 0$ is company $r$'s investment at time $t$ on product $H_r$'s quality\\ 
$\bm{c}$ & the budget vector. $\bm{c}\in \mathbb{R}^{R\times 1}$ and $\bm{c}\succ \vectorzeros[R]$. entry $c_r$ is the budget limit for company $r$\\
$\psi_r(\bm{x}^{(i)};\gamma)$ & $\psi_r:\mathbb{R}_{\ge 0}^{1\times R}\rightarrow \mathbb{R}_{\ge 0}$ defined by $\psi_r(\bm{x}^{(i)};\gamma)=x_{ir}/(\bm{x}^{(i)}\vectorones[R]+\gamma)$, with model parameter $\gamma >0$\\
$g_r(\bm{w})$ & $g_r:\mathbb{R}_{\ge 0}^{R\times 1} \rightarrow \mathbb{R}_{\ge 0}$ defined by $g_r(\bm{w}) = w_r/\vectorones[R]^{\top}\bm{w}$, if $w_r>0$; $g_r(\bm{w})=0$, if $w_r=0$\\
$\Omega$ & the set of all the players' possible actions. $\Omega = \{(X,\bm{w})| X\succeq \vectorzeros[n\times R],\bm{w}\succeq \vectorzeros[R], \vectorones[n]^{\top}X+\bm{w}^{\top}\preceq \bm{c}^{\top}\}$ \\
$\Omega_r$ & the action space for company $r$. $\Omega_r=\{(\bm{x}_r,w_r) \,|\, \bm{x}_r\succeq \vectorzeros[R],w_r\ge 0, \vectorones[n]^{\top}\bm{x}_r+w_r\le c_r\}$  \\
$\bm{\beta}_r(t)$ & $\bm{\beta}_r(t)=\big( \beta_{1r}(t),\dots,\beta_{nr}(t) \big)^{\top}=\tilde{A}\bm{p}_r(t)$\\
\bottomrule
\end{tabular}
\end{center}
\end{table} 
 
\subsection{Model 1: competitive seeding-quality game}
\emph{a) Model set-up:} 
The multi-stage competitive seeding-quality game is formalized as follows.

\emph{a.1) Players:} The players are the $R$ companies. Each company $r$ has a product $H_r$ competing on the network.

\emph{a.2) Players' actions:} At each time step, each company $r$ has two types of investments. The investment on seeding, i.e., $\bm{x}_r(t)$, changes the individuals' product-adoption probability in the social conversion process, while the investment on quality, i.e., $w_r(t)$, influences the product-conversion graph. The total investment is bounded by a fixed budget $c_r$, i.e., $\vectorones[n]^{\top}x_r(t)+w_r(t)\le c_r$. 

\emph{a.3) Rules:} We model the effect of investment on seeding as follows. For any individual $i\in V$, each company $r$'s investment $x_{ir}(t)$ creates a "virtual node" in the network, who is always adopting the product $H_r$. 
In the social conversion process, the probability that individual $i$ picks any company $r$'s virtual node is $\psi_r\big( \bm{x}^{(i)}(t);\gamma \big)$ for any $i\in V$ and $r\in \Theta$. The probability that individual $i$ picks individual $j$ in the social conversion process is then given by $\Big( 1-\sum_{s=1}^R \psi_s\big( \bm{x}^{(i)};\gamma \big) \Big)\tilde{a}_{ij}$. As for the investment on product quality, we assume that the product-conversion graph is associated with a rank-one adjacency matrix $[\delta_1 \vectorones[n],\delta_2 \vectorones[n],\dots,\delta_R \vectorones[n]]$ and $\delta_r=g_r(\bm{w}(t))$ is determined by all the companies' investments on product quality. We exclude the non-realistic case when $\bm{w}=\vectorzeros[R]$ by imposing an additional condition on the budget vector $\bm{c}$, which is specified later in this subsection. With each company $r$'s action $\big( \bm{x}_r(t),w_r(t) \big)$ at time $t$, the dynamics of the product-adoption probabilities $P(t)\in \mathbb{R}^{n\times R}_{\ge 0}$ is given by
\begin{equation}\label{eq_game_dyn_seeding_quality}
\begin{split}
p_{ir}(t)&= \alpha_i \frac{\gamma}{\bm{x}^{(i)}(t)\vectorones[R]+\gamma} \sum_{k=1}^n \tilde{a}_{ik} p_{kr}(t-1)\\
&\quad + \alpha_i\psi_r \big( \bm{x}^{(i)}(t);\gamma \big)+(1-\alpha_i)g_r\big( \bm{w}(t) \big),
\end{split}
\end{equation} 
for any $i\in V$ and $r\in \Theta$.

\emph{a.4) Pay-offs:} At each time step $t$, the pay-off for any company $r$, denoted by $u_r\big( X(t),\bm{w}(t) \big)$, is the total adoption probability of $H_r$, i.e., $u_r\big( X(t),\bm{w}(t) \big)=\vectorones[n]^{\top}\bm{p}_r(t)$, where $\bm{p}_r(t)$ is the $r$-th column vector of $P(t)$.

Any Nash equilibrium $\big(X^*(t),\bm{w}^*(t)\big)$ at stage $t$ for the multi-stage game defined above satisfies that, for any $r$,
\begin{equation}\label{eq_game_def_Nash}
u_r\big( X^*(t),\bm{w}^*(t) \big) \ge u_r\big( X(t),\bm{w}(t) \big),
\end{equation}
for any $\big( X(t),\bm{w}(t) \big)\in \Omega$ such that $\bm{x}_s(t)=\bm{x}_s^*(t)$ and $w_s(t)=w_s^*(t)$ for any $s\neq r$, and $\bm{x}_r(t)\neq \bm{x}_r^*(t)$ and $w_r(t)\neq w_r^*(t)$.

\emph{b) Theoretical Analysis of Nash equilibrium and system's dynamical behavior:}
The following theorem gives a closed-form expression of the Nash equilibrium at each stage and the system's asymptotic behavior when every player is adopting the policy at the Nash equilibrium.
\begin{theorem}[Competitive seeding-quality game]\label{thm_game_seeding_quality}
Consider the $R$-player multi-stage game described above in this subsection. Further assume that the budget limit $c_r$ for any company $r$ satisfies 
\begin{equation}\label{eq_game_seeding_quality_condition_budget}
c_r>\max\Big\{n\gamma \frac{\vectorones[n]^{\top}(\vectorones[n]-\bm{\alpha})}{\vectorones[n]^{\top}\bm{\alpha}}, (\frac{n}{\min_i \alpha_i}-1)\gamma\Big\}.
\end{equation}
Then we have the following conclusions: 

i) at each time step $t+1$, the Nash equilibrium $\big( X^*(t+1),\bm{w}^*(t+1) \big)$ is given by
\begin{align}
\label{eq_game_seeding_quality_Nash_eq_x} x_{ir}^*(t+1)&=\frac{\alpha_i}{n}c_r + \frac{\alpha_i \gamma}{n}\vectorones[n]^{\top}\bm{\beta}_r(t) - \beta_{ir}(t)\gamma,\\
\label{eq_game_seeding_quality_Nash_eq_w} w_r^*(t+1)&=\left(1-\frac{\vectorones[n]^{\top}\bm{\alpha}}{n}\right)\big( c_r+\vectorones[n]^{\top}\bm{\beta}_r(t) \gamma \big),
\end{align}
and $x^*_{ir}(t)>0$, $w^*_r(t)>0$ for any $i\in V,r\in \Theta$; 

ii) if $\big(X(t),\bm{w}(t)\big)=\big( X^*(t),\bm{w}^*(t) \big)$ for
any $t\in \mathbb{N}$ and $P(0)\in S_{nR}(\vectorones[n])$, then
$P(t)$, for $t\ge 1$, is a rank-one matrix of the form $\big[
  p_1(t)\vectorones[n],p_2(t)\vectorones[n],\dots,p_R(t)\vectorones[n]
  \big]$. Moreover, for any $r\in \Theta$, as $t\rightarrow \infty$, $p_r(t)$ converges to $c_r/\vectorones[R]^{\top}\bm{c}$ exponentially fast with the rate $n\gamma / (\vectorones[R]^{\top}\bm{c}+n\gamma)$.
\end{theorem}

\begin{proof}
Since we only discuss time step $t+1$ in this proof, for simplicity of notations and without causing any confusion, we use $x_{ir}$ ($w_r$, $x^*_{ir}$, $w^*_r$ resp.) for $x_{ir}(t+1)$ ($w_r(t+1)$, $x^*_{ir}(t+1)$, $w^*_r(t+1)$ resp.). 

If company $r$ knows the actions of all the other companies at time step $t+1$, i.e., $(\bm{x}_s,w_s)$, for any $s\neq r$, the optimal response for company $r$ is the solution to the following optimization problem: 
\begin{equation}\label{problem_seeding_quality_optimization}
\begin{aligned}
& \underset{(\bm{x},w)\in \Omega_r}{\text{minimize}}
&& -\vectorones[n]^{\top}\bm{p}_r(t+1)\\
& \text{subject to}
&& \vectorones[n]^{\top}\bm{x}+w-c_r\le 0.
\end{aligned}
\end{equation}
Let $\tilde{x}_{ir}=x_{ir}+\beta_{ir}(t)\gamma$ and $L_r(\bm{x}_r,w_r,\mu_r)=-\vectorones[n]^{\top}\bm{p}_r(t+1)+\mu_r \vectorones[n]^{\top}\bm{x}_r+\mu_r w_r - \mu_r c_r$, for any $i\in V$ and $r\in \Theta$. The solution to the optimization problem~\eqref{problem_seeding_quality_optimization} satisfies
\begin{align}
\label{eq_game_seeding_quality_partial_L_x} \frac{\partial L_r}{\partial x_{ir}} & = -\alpha_i \frac{\sum_{s\neq r}\tilde{x}_{is}}{(\sum_{s=1}^R \tilde{x}_{is})^2} + \mu_r =0,\\
\label{eq_game_seeding_quality_partial_L_w} \frac{\partial L_r}{\partial w_r} & = -\vectorones[n]^{\top}(\vectorones[n]-\bm{\alpha})\frac{\sum_{s\neq r}w_s}{(\vectorones[R]^{\top}\bm{w})^2} + \mu_r = 0,\\
\label{eq_game_seeding_quality_partial_L_mu} \frac{\partial L_r}{\partial \mu_r} & = \vectorones[n]^{\top}\bm{x}_r + w_r -c_r = 0.
\end{align}

According to equation~\eqref{eq_game_def_Nash}, $(\bm{x}^*_r,w^*_r)$ is the solution to the optimization problem~\eqref{problem_seeding_quality_optimization} with $(\bm{x}_s,w_s)=(\bm{x}^*_s,w^*_s)$ for any $s\neq r$. One immediate result is that $\vectorones[n]^{\top}\bm{x}^*_r + w^*_r -c_r = 0$ for any $r\in \Theta$. Then, according to equation~\eqref{eq_game_seeding_quality_partial_L_x}, we have that, at the Nash equilibrium at time step $t+1$,
\begin{equation*}
\frac{1}{\sqrt{\mu_r}} = \frac{1}{\sum_{k=1}^n \sqrt{\alpha_k \sum_{s\neq r}\tilde{x}^*_{ks}}} \sum_{s=1}^R \big( c_s-w^*_s+\vectorones[n]^{\top}\bm{\beta}_s \gamma \big),
\end{equation*} 
and therefore,
\begin{equation}\label{eq_game_seeding_quality_proof_1}
\frac{\sqrt{\alpha_i \sum_{s\neq r}\tilde{x}_{is}^*}}{\sum\limits_{k=1}^n \sqrt{\alpha_k \sum\limits_{s\neq r}\tilde{x}_{ks}^*}} = \frac{\sum\limits_{s=1}^R \tilde{x}_{is}^*}{\sum\limits_{s=1}^R \big( c_s-w_s^*+\vectorones[n]^{\top}\bm{\beta}_s(t) \gamma \big)}.
\end{equation}
The right-hand side of the equation above does not depend on the product index $r$. Therefore,
\begin{equation*}
\frac{\sqrt{\alpha_i \sum_{s\neq r}\tilde{x}_{is}^*}}{\sum_{k=1}^n \sqrt{\alpha_k \sum_{s\neq r}\tilde{x}_{ks}^*}} = \frac{\sqrt{\alpha_i \sum_{s\neq \tau}\tilde{x}_{is}^*}}{\sum_{k=1}^n \sqrt{\alpha_k \sum_{s\neq \tau}\tilde{x}_{ks}^*}},
\end{equation*}
for any $r,\tau \in \Theta$. The equation above leads to
\begin{equation*}
\frac{\sum_{s\neq r}\tilde{x}^*_{is}}{\sum_{s\neq \tau}\tilde{x}^*_{is}} = \left( \frac{\sum_{k=1}^n \sqrt{\alpha_k \sum_{s\neq r} \tilde{x}_{ks}^*}}{\sum_{k=1}^n \sqrt{\alpha_k \sum_{s\neq \tau} \tilde{x}_{ks}^*}} \right)^2.
\end{equation*}
Since the right-hand side of the equation above does not depend on the individual index $i$, we have 
\begin{equation*}
\frac{\sum_{s\neq r}\tilde{x}^*_{is}}{\sum_{s\neq r} \tilde{x}^*_{js}} = \frac{\sum_{s\neq \tau}\tilde{x}^*_{is}}{\sum_{s\neq \tau} \tilde{x}^*_{js}} = 
\frac{\sum_{s=1}^R \tilde{x}^*_{is}}{\sum_{s=1}^R \tilde{x}^*_{js}} = 
\frac{\tilde{x}^*_{ir}}{\tilde{x}^*_{jr}},
\end{equation*}
for any $r,\tau \in \Theta$. Combine the equation above with equation~\eqref{eq_game_seeding_quality_proof_1} and then we obtain
\begin{equation*}
\frac{\sum_{s=1}^R \tilde{x}^*_{is}}{\sum_{s=1}^R \tilde{x}^*_{js}} = \sqrt{\frac{\alpha_i}{\alpha_j}} \sqrt{\frac{\sum_{s\neq r}\tilde{x}^*_{is}}{\sum_{s\neq r}\tilde{x}^*_{js}}}\quad \Rightarrow \quad \frac{\tilde{x}^*_{ir}}{\tilde{x}^*_{jr}} = \frac{\alpha_i}{\alpha_j},
\end{equation*}
for any $r\in \Theta$. Therefore,
\begin{equation}\label{eq_game_seeding_quality_proof_2}
\tilde{x}^*_{ir} = \frac{\alpha_i}{\vectorones[n]^{\top}\bm{\alpha}} \big( c_r-w^*_r + \vectorones[n]^{\top}\bm{\beta}_r(t)\gamma \big).
\end{equation}

Combining equation~\eqref{eq_game_seeding_quality_proof_2} and~\eqref{eq_game_seeding_quality_partial_L_w}, we obtain
\begin{equation*}
\frac{c_r-w^*_r+\vectorones[n]^{\top}\bm{\beta_r}(t)\gamma}{w^*_r}=\frac{c_{\tau}-w^*_{\tau}+\vectorones[n]^{\top}\bm{\beta}_{\tau}(t)\gamma}{w^*_{\tau}} = \eta,
\end{equation*}
for any $r,\tau \in \Theta$ and some constant $\eta$. Substitute the equation above back into equation~\eqref{eq_game_seeding_quality_partial_L_w}, we solve that $\eta = \vectorones[n]^{\top}\bm{\alpha}/\vectorones[n]^{\top}(\vectorones[n]-\bm{\alpha})$. Therefore, we obtain equation~\eqref{eq_game_seeding_quality_Nash_eq_w} and by substituting equation~\eqref{eq_game_seeding_quality_Nash_eq_w} into equation~\eqref{eq_game_seeding_quality_proof_2} we obtain equation~\eqref{eq_game_seeding_quality_Nash_eq_x}. Moreover, one can check that equation~\eqref{eq_game_seeding_quality_condition_budget} guarantees $\tilde{x}^*_{ir}>0$ and $w^*_r>0$ for any $i\in V$ and $r\in \Theta$. This concludes the proof for Conclusion i).

Substituting euqation~\eqref{eq_game_seeding_quality_Nash_eq_x} and~\eqref{eq_game_seeding_quality_Nash_eq_w} into the dynamical system~\eqref{eq_game_dyn_seeding_quality}, after simplification, we have
\begin{equation*}
\bm{p}_r(t+1) = \frac{c_r+\vectorones[n]^{\top}\tilde{A}\bm{p}_r(t)\gamma}{\vectorones[R]^{\top}\bm{c}+n\gamma} \vectorones[n] = p_r(t+1) \vectorones[n].
\end{equation*}
Therefore, $p_r(t+1) = \big(c_r+n\gamma p_r(t)\big)/(\vectorones[R]^{\top}\bm{c}+n \gamma)$ for any $t\ge 1$.
One can check that the equation above leads to all the results in Conclusion ii). 
\end{proof}

\emph{c) Interpretations and Remarks: }The basic idea of seeding-quality trade-off in the competitive seeding-quality game is similar to the work by Fazeli et. al.~\cite{AF-AA-AJ:15} but our model is essentially different from~\cite{AF-AA-AJ:15} in that our model is multi-stage, and takes the products, rather than the individuals, as the players. Moreover, our model is based on a different network propagation dynamics.

Theorem~\ref{thm_game_seeding_quality} provides some strategic insights on the investment decisions and the seeding-quality trade-off.

\emph{c.1) Interpretation of $\beta_{ir}(t)$:} By definition, $\beta_{ir}(t)$ is the average probability, among all the neighbors of individual $i$, of adopting product $H_r$ at time step $t$, while $\vectorones[n]^{\top}\bm{\beta}_r(t)/n=\sum_{l=1}^n \sum_{k=1}^n \tilde{a}_{lk}\beta_{kr}(t)/n$ is a convext combination of all the entries of $\bm{\beta}_r(t)$ and characterizes the current overall acceptance of product $H_r$.

\emph{c.2) Seeding-quality trade-off:} According to equation~\eqref{eq_game_seeding_quality_Nash_eq_w}, at the Nash equilibrium, the investment on $H_r$'s product quality monotonically decreases with $\vectorones[n]^{\top}\bm{\alpha}/n$, and increases with $\vectorones[n]^{\top}\bm{\beta}_r$. 
This observation implies that: 1) in a society with relatively low open-mindedness, the competing companies should relatively emphasize more on improving their products' quality, rather than seeding, and vice versa; 2) for products which are currently not widely adopted, seeding is relatively more efficient than improving the product's quality.

\emph{c.3) Allocation of seeding resources among the individuals:} According to equation~\eqref{eq_game_seeding_quality_Nash_eq_x}, for any company $r$, at the Nash equilibrium at each time step $t+1$, the investment on seeding for any individual $i$, i.e., $x_{ir}(t+1)$, increases with individual $i$'s open-mindedness, since it is easier for a more open-minded individual to be influence by seeding. Moreover, if we rewrite equation~\eqref{eq_game_seeding_quality_Nash_eq_x} as
\begin{equation*}
x^*_{ir}(t+1) = \frac{\alpha_i}{n} c_r + \frac{\alpha_i \gamma}{n} \sum_{k\neq i} \beta_{kr}(t) - (1-\frac{\alpha_i}{n})\gamma\beta_{ir}(t),
\end{equation*}  
one would observe that $x^*_{ir}(t)$ monotonically decreases with $\beta_{ir}(t)$, which denotes the average probability of adopting $H_r$ among individual $i$'s neighbors. A possible explanation is that, with large $\beta_{ir}(t)$, individual $i$ is very likely to be converted to $H_r$ due to her neighbors, and thereby the seeding for individual $i$ is relatively not necessary. Once again we rewrite equation~\eqref{eq_game_seeding_quality_Nash_eq_x} as
\begin{equation*}
\begin{split}
x^*_{ir}(t+1) & = \frac{\alpha_i}{n}c_r + \frac{\alpha_i \gamma}{n}\sum_{l=1}^n \tilde{a}_{li}p_{ir}(t)\\
&\quad + \frac{\alpha_i \gamma}{n} \sum_{l=1}^n \sum_{k\neq i}\tilde{a}_{lk}p_{kl}(t) - \gamma \sum_{k=1}^n \tilde{a}_{ik}p_{kr}(t).
\end{split}
\end{equation*}  
On the right-hand side of the equation above, only the second term contains $p_{ir}(t)$. Therefore, $x^*_{ir}(t)$ increases with $\sum_{l=1}^n \tilde{a}_{li}p_{ir}(t)$, in which $\sum_{l=1}^n \tilde{a}_{li}$ is individual $i$'s in-degree, reflecting $i$'s potential of influencing the others, and $\sum_{l=1}^n \tilde{a}_{li}p_{ir}(t)$ characterizes individual $i$'s potential of converting other individuals to product $H_r$. 

\emph{c.4) Nash equilibrium on the boundary:} Without equation~\eqref{eq_game_seeding_quality_condition_budget}, the right-hand sides of equation~\eqref{eq_game_seeding_quality_Nash_eq_x} and~\eqref{eq_game_seeding_quality_Nash_eq_w} could be non-positive. In this case, the analysis becomes more complicated and the Nash equilibrium would be on the boundary of $\Omega$, i.e., some of the $x^*_{ir}(t)$ or $w^*_r(t)$ should be 0. 

\emph{c.5 Preset quality:} In order to make the quality function $g_r(\bm{w})$ smooth at $\bm{w}=\vectorzeros[R]$, we can modify its definition as $g_r(\bm{w};\bm{\xi},u)=(w_r+\xi_r u)/(\vectorones[R]^T\bm{w}+u)$, where $u>0$ and $\bm{xi}\succ \vectorzeros[R]$ are model parameters and $\vectorones[R]^T \bm{\xi}=1$. The parameter $\xi_r$ characterizes product $H_r$'s preset relative quality. Based on the same argument of the proof for Theorem~\ref{thm_game_seeding_quality}, one can check that, for the Nash equilibrium at which all the individuals' investments on both seeding and quality are positive, $w_r^*(t+1)+u\vectorones[n]^{\top}\bm{\alpha}\xi_r/n$ is equal to the right-hand side of equation~\eqref{eq_game_seeding_quality_Nash_eq_w} and thereby any company $r$'s investment on quality decreases with $H_r$'s preset relative quality.   


\subsection{Model 2: competitive seeding game}
In this subsection we consider the case when the products' relative quality is fixed and the competing companies can only invest on seeding. The model set-up is the same with the game proposed in Section V.A, except that the action for any company $r$ is $\bm{x}_r\in \mathbb{R}^{n\times 1}_{\ge 0}$, constrained by $\vectorones[n]^{\top}\bm{x}_r\le c_r$. The dynamics of $P(t)=\big( p_{ir}(t) \big)_{n\times R}$, with seeding actions $X(t)$, is given by
\begin{equation}\label{eq_game_seeding_only_dyn_sys}
\begin{split}
p_{ir}(t+1) &= \alpha_i \frac{\gamma\beta_{ir}(t)+x_{ir}(t)}{\bm{x}^{(i)}(t)\vectorones[R]+\gamma}\\
&\quad + (1-\alpha_i)\sum_{s=1}^R p_{is}(t)\delta_{sr}.
\end{split}
\end{equation}

The following theorem gives the closed-form expression of the unique Nash equilibrium at each time step and the system's dynamical behavior if all the companies are adopting the policies at the Nash equilibrium.
\begin{theorem}[Competitive propagation game with seeding only]\label{thm_game_seeding_only}
Consider the competitive seeding game described in this subsection. Assume that $c_r>(\vectorones[n]^{\top}\bm{\alpha}/\min_i \alpha_i - 1)\gamma$ for any $r\in \Theta$. Then we have:

i) at each time step, the Nash equilibrium $X^*(t)$ is given by
\begin{equation*}
\bm{x}^*_r(t) = \frac{c_r+\vectorones[n]^{\top}\bm{\beta}_r(t)\gamma}{\vectorones[n]^{\top}\bm{\alpha}} \bm{\alpha} - \gamma \bm{\beta}_r(t),
\end{equation*}
and $\bm{x}^*_r(t)\succ \vectorzeros[n]$, for any $r\in \Theta$;

ii) if $X(t)=X^*(t)$ for any $t\in \mathbb{N}$, then the matrix form of system~\eqref{eq_game_seeding_only_dyn_sys} is given by 
\begin{equation}\label{eq_game_seeding_only_Nash_dyn}
\begin{split}
P(t+1)&=\diag(\bm{\alpha})\frac{\vectorones[n]\bm{c}^{\top}+\gamma (\vectorones[n]\vectorones[n]^{\top})\tilde{A}P(t)}{\vectorones[n]^{\top}\bm{c}+n\gamma}\\
&\quad + (I-\diag(\bm{\alpha}))P(t)\Delta.
\end{split}
\end{equation}
Moreover, there exists a unique fixed point $P^*\in S_{nR}(\vectorones[n])$ for the system above, and, for any $P(0)\in S_{nR}(\vectorones[n])$, $P(t)$ converges to $P^*$ exponentially fast with the convergence rate $\eta = \max_i \alpha_i n \gamma /(\vectorones[n]^{\top}\bm{c}+n\gamma)+(1-\alpha_i)\zeta(\Delta)$, 
where $0\le \zeta(\Delta)=1-\sum_{r=1}^R \min_s \delta_{sr}\le 1$.
\end{theorem}

Proof of Conclusion i) in Theorem~\ref{thm_game_seeding_only} is similar to the proof of Conclusion i) in Theorem~\ref{thm_game_seeding_quality}. As for Conclusion ii) in Theorem~\ref{thm_game_seeding_only}, following the same line of argument in the proof of Case 1 in Theorem~\ref{thm_asym_behav_social_self_NCPM} and applying the Banach fixed point theorem, one can check that equation~\eqref{eq_game_seeding_only_Nash_dyn} defines a contraction map. 

Interpretations on the allocation of seeding investments among the individuals, for the competitive seeding-quality game proposed in the previous subsection, also applies to the competitive seeding game proposed in this subsection. One main difference in dynamical property between these two game-theoretic models is that, in the competitive seeding game defined in Section V.B, the matrix $P(t)$ does not become a rank-one matrix with all the players adopting the policies at the Nash equilibrium, and moreover, the fixed point $P^*$ for system~\eqref{eq_game_seeding_only_Nash_dyn} depends not only on the budgets vector $\bm{c}$, but also on $\bm{\alpha}$ and $\tilde{A}$.

\subsection{Simulation illustration}
Simulation work has been conducted for the two game-theoretic models analyzed in this section, on a competitive propagation system with $R=2$, $n=5$ and $\tilde{A}$ being strongly connected. The model parameters are set as $\bm{\alpha}=(0.8,0.85,0.75,0.84,0.76)^{\top}$, $\gamma = 100$, $c_1=600$, $c_2=900$, such that the conditions on $c_r$ in Theorem~\ref{thm_game_seeding_quality} and~\ref{thm_game_seeding_only} are satisfied. For the competitive seeding game, $\Delta$ is randomly generated.
Figure~\ref{fig_game_seeding_quality} shows that, for the competitive seeding-quality game, the average probability of adopting $H_1$ ($H_2$ resp.) converges to $c_1/(c_1+c_2)$ ($c_2/(c_1+c_2)$ resp.), as predicted by Theorem~\ref{thm_game_seeding_quality}. In addition, Figure~\ref{fig_game_seeding_only} shows that, for the competitive seeding game, the average probability of adopting $H_1$ (or $H_2$ resp.) also converges. Moreover, one can observe, from both Figure~\ref{fig_game_seeding_quality} and~\ref{fig_game_seeding_only}, that the company adopting the Nash policy, with the other company just randomly allocating the investment, has higher pay-off than in the case in which both sides are adopting the policies at the Nash equilibrium.

\begin{figure}
\begin{center}
\subfigure[competitive seeding-quality game]{\label{fig_game_seeding_quality} \includegraphics[width=.46\linewidth]{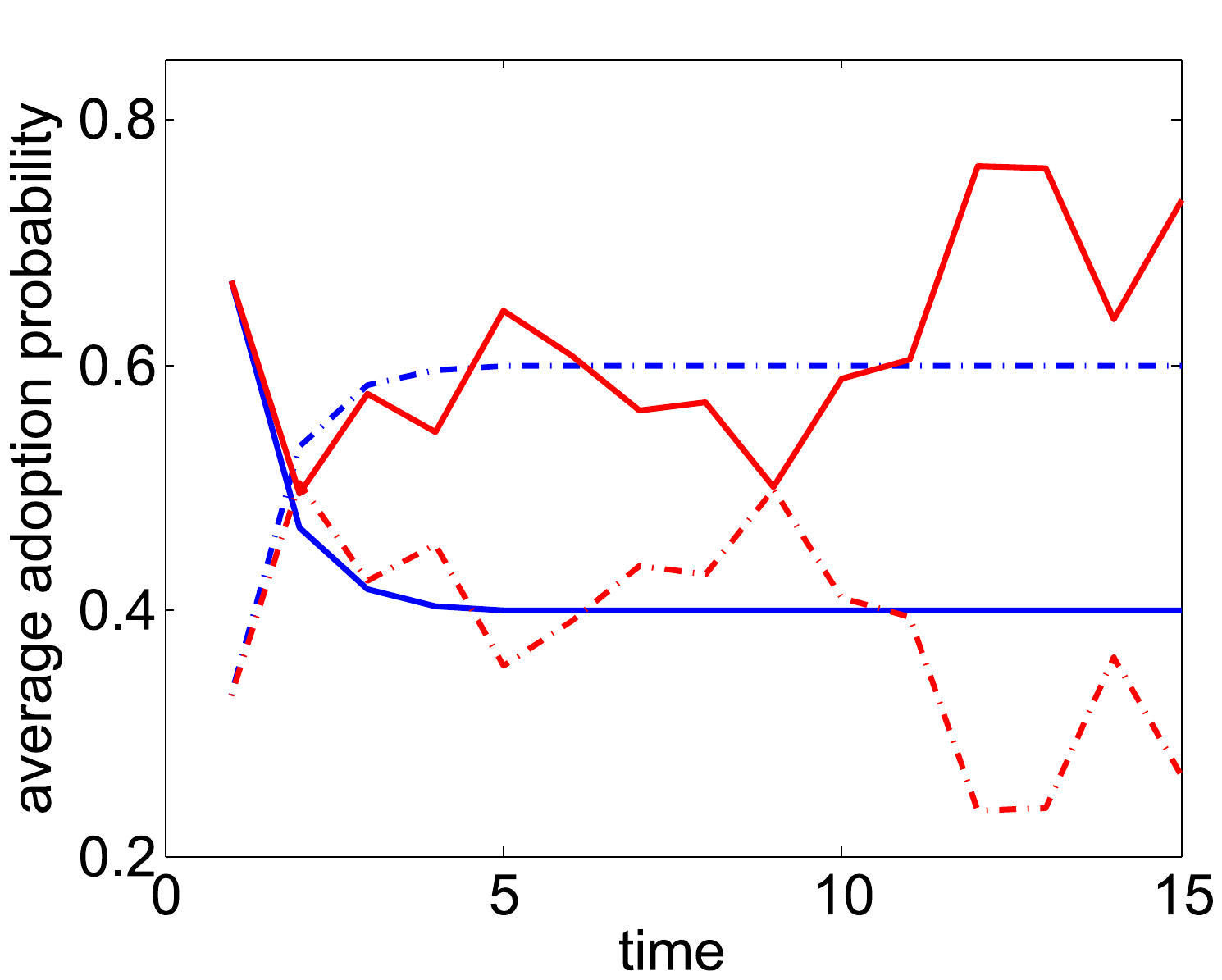}}
\subfigure[competitive seeding game]{\label{fig_game_seeding_only} \includegraphics[width=.46\linewidth]{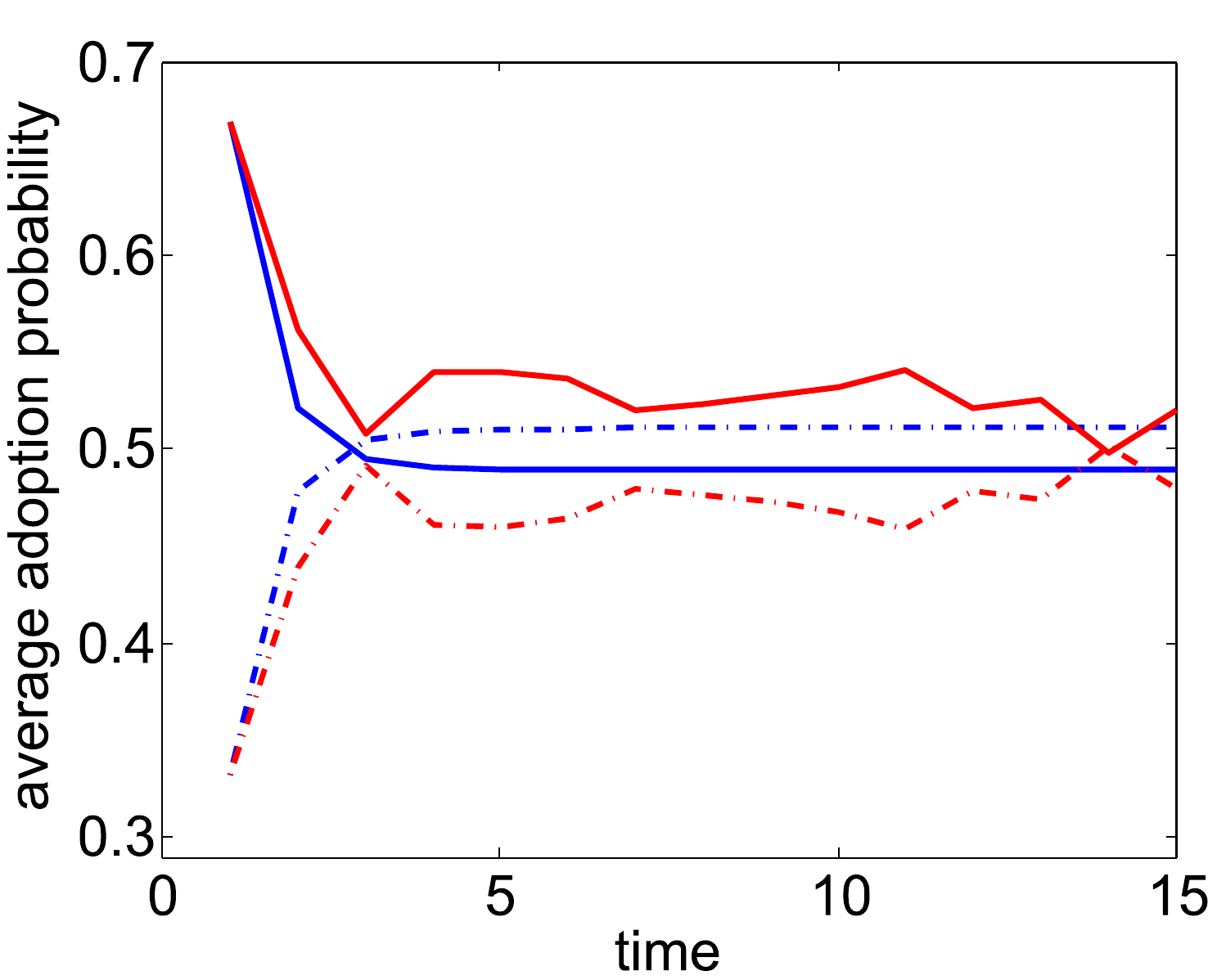}}
\caption{Evolution of the average adoption probabilities $\vectorones[n]^{\top}\bm{p}_r(t)/n$ for the two game models proposed in Section V. The solid curves represent the average adoption probability for product $H_1$, while the dash curves represents the average adoption probability for product $H_2$. The blue curves correspond to the case in which both companies are adopting the Nash policies, while the red curves correspond to the case when company 1 is adopting the Nash policy but company 2 is randomly allocating the investment on seeding and product quality.}
\label{fig_game}
\end{center}
\end{figure}

\section{Conclusion}
This paper discusses a class of competitive propagation models based on two product-adoption mechanisms: the social conversion and the self conversion. Applying the independence approximation we propose two difference equations systems, referred to as the social-self NCPM and the self-social NCPM respectively. Theoretical analysis reveals that the structure of the product-conversion graph plays an important role in determining the nodes' asymptotic state probability distributions. Simulation results reveal the high accuracy of the independence approximation and the asymptotic behavior of the original social-self Markov chain model. Based on the social-self NCPM, we propose a class of multi-stage competitive propagation games and discuss the trade-off between seeding and quality at the unique Nash equilibrium. One possible future work is the deliberative investigation on the Nash equilibrium on the boundary. Another open problem is the stability analysis of the self-social NCPM with $R>2$. Simulation results support the claim that, for the self-social NCPM with $R>2$, there also exists a unique fixed point $P^*$ and, for any initial condition $P(0)\in S_{nR}(\vectorones[n])$, the solution $P(t)$ to equation~\eqref{eq_mf_self_social_matrix} converges to $P^*$. We leave this statement as a conjecture.

\bibliographystyle{plain}
\bibliography{alias,FB,Main}

\end{document}